\documentclass[twocolumn,10pt]{IEEEtran}
\usepackage{graphicx}
\usepackage{float}
\usepackage{amsfonts}
\usepackage{amsmath}
\usepackage{amssymb}
\usepackage{cite}
\usepackage{bm}
\usepackage{mathrsfs}
\usepackage{epsfig}
\usepackage{algorithm}
\usepackage{algorithmic}
\usepackage{amsthm}
\usepackage{color}
\usepackage{subfigure}

\newtheorem{theorem}{Theorem}
\newtheorem{remarks}{Remarks}
\newtheorem{lemma}{Lemma}

\newtheorem{proposition}{Proposition}
\newtheorem{corollary}{Corollary}
\newtheorem{definition}{Definition}

\newcommand{\inner}[2]{{\langle#1,#2\rangle}}
\begin{document}
\title{Subspace Recovery  from  Structured Union of Subspaces}
\author{\authorblockA{Thakshila Wimalajeewa$^{\dag}$, \emph{Member IEEE}, Yonina C. Eldar$^{\dag\dag}$, \emph{Fellow IEEE},  and Pramod K.
Varshney$^{\dag}$, \emph{Fellow IEEE}}
}

\maketitle \thispagestyle{empty}

\begin{abstract}
Lower dimensional signal representation schemes  frequently  assume that the signal of interest  lies in a single vector space. In  the context of the recently developed theory of compressive  sensing (CS), it is often assumed that the signal of interest is sparse in an orthonormal basis.  However, in many practical applications,  this requirement may be too restrictive. A generalization of the standard sparsity assumption is that the signal lies in a union of subspaces. Recovery of such signals from a small number of samples has been studied recently in several works. Here, we consider the problem of subspace recovery  in which our goal is to identify the subspace (from the union) in which the signal lies using  a small number of samples, in the presence of noise.
More specifically, we derive performance bounds and  conditions under which  reliable subspace  recovery   is guaranteed using maximum likelihood (ML) estimation.  We begin by treating general unions and then obtain  the results for  the special case in which the subspaces  have structure leading to block sparsity. In our analysis, we treat both general sampling operators  and random sampling matrices. With general unions, we show that under certain conditions, the number of measurements required for reliable subspace recovery  in the presence of noise via ML   is less than that implied using the restricted isometry property which guarantees signal recovery. In the special case of  block sparse signals,  we quantify the gain achievable  over standard sparsity in subspace recovery. Our results also strengthen existing results on  sparse support  recovery in the presence of noise under  the standard sparsity model.

\end{abstract}

{\bf\emph{ Index terms}- Maximum likelihood estimation, union of linear subspaces, subspace recovery,   compressive sensing,  block sparsity }

\footnotetext[1]{$^{\dag}$Dept. of Electrical Engineering  and Computer Science,
Syracuse University, Syracuse, NY 13244\\
$^{\dag\dag}$Dept. of Electrical Engineering, Technion-Israel Institute of Technology, Technion City, Haifa 32000, Israel\\
Email: twwewelw@syr.edu, yonina@ee.technion.ac.il, varshney@syr.edu}
\footnotetext[2]{The work of T. Wimalajeewa and P. K. Varshney  was supported by the National Science
Foundation (NSF) under Grant No. 1307775. The work of Y. C. Eldar was supported in part by the Israel Science Foundation under Grant no. 170/10, in part by the SRC, in part by the Ollendorf Foundation, and in part by the Intel Collaborative Research Institute for Computational Intelligence (ICRI-CI).}

\section{Introduction}
The compressive  sensing (CS) framework has established that a small number of measurements acquired via random projections are  sufficient for signal recovery when the signal of interest is sparse in a certain basis.  Consider a length-$N$  signal  $\mathbf x$ which  can be represented in a basis $\mathbf V$ such that $\mathbf x = \mathbf V \mathbf c$. The signal $\mathbf x$ is said to be $k$-sparse in the basis $\mathbf V$ if  $\mathbf c$ has only $k $ nonzero coefficients where $k$ is much smaller than $N$. It  has been shown in \cite{candes1,donoho1,candes2} that $\mathcal O( k \log (N/k))$  compressive measurements are sufficient to recover $\mathbf x$
when the elements of the measurement matrix  are random. Signal recovery  can be performed via  optimization or greedy based approaches. A detailed overview of CS can be found in \cite{Eldar_B1}.

There are  a variety of  applications in which complete signal recovery is not necessary.  The problem of
sparse support   recovery (equivalently sparsity pattern recovery or  finding the locations of nonzero coefficients  of a sparse signal)  arises in a wide variety of areas including source localization \cite{Malioutov1,Cevher1},
sparse approximation \cite{Natarajan1}, subset selection in linear regression \cite{Miller1,Larsson1}, estimation of frequency band locations in cognitive radio networks \cite{Tian1,Mishali2,Mishali3},
and signal denoising \cite{SSChen1}. In these applications, often finding the
sparsity pattern  of the signal is more important than approximating the signal itself. Further, in the CS framework, once the  sparse support is  identified, the signal can be estimated using standard techniques. For  the problem of complete sparse signal recovery, there is a significant amount of work in the literature that focuses on deriving  recovery guarantees and stability with respect to various $l_q$ norms of the reconstruction error. However, as pointed out in \cite{wain2}, recovery guarantees derived for sparse signals do not always imply  exact recovery of the sparse support. The criteria used for sparse support  recovery  and exact signal recovery  are generally  different. Although a signal estimate can be close to the original sparse signal, the estimated signal  may  have a different support compared to the true signal support   \cite{wain2}.   For example, Lasso has been shown to be information theoretically optimal in certain regimes of the signal-to-noise ratio (SNR) for sparse support recovery, while in other regimes of SNR, the Lasso fails with high probability in recovering the sparsity pattern \cite{wain2,Wainwright4}.  Thus,  investigation of  recovery conditions for sparse support   at any given SNR  is an important problem. Performance  limits on
reliable recovery of the sparsity pattern
have been derived by several authors in  recent
work exploiting
information theoretic tools \cite{wain2,wang5,Fletcher1,Akcakaya1,Reeves1, tang1,goyal1,thakshilaj4}.
Most of these works focus on deriving  necessary and sufficient conditions
for reliable sparsity pattern  recovery assuming  the  standard sparsity model.

There are practical scenarios where structured properties of the signal are available. Reduced dimensional signal processing for  several signal models which go beyond  simple sparsity  has been treated  in recent literature \cite{Baraniuk4,Lu1,Eldar1,Blumensath1,Eldar3,Duarte2}.
One general model that can describe many structured problems is that of a union of subspaces. In this setting, the signal is known to lie in one out of a possible set of subspaces but the specific subspace chosen is unknown. Examples include  wideband spectrum sensing \cite{Mishali2},  time delay estimation with overlapping echoes \cite{Bruchstein1,Lu1,Gedalyahu1},  and signals having finite rate innovation \cite{Ben1,Dragotti1}.
Conditions under which stable sampling and recovery is possible in a general union of subspaces model
are derived in \cite{Lu1,Eldar1,Blumensath1,Baraniuk4}.  However, the problem of recovering the subspace in which the signal lies without completely recovering the signal (or the problem of subspace recovery) has not been treated  in this more general setting.

In this paper, our goal is to investigate the problem of subspace recovery in  the union of subspaces model with a given sampling operator.  We  consider  subspace recovery based on the optimal ML decoding scheme in the presence of noise. While ML  is computationally intractable as the signal dimension increases, the analysis gives a benchmark for the optimal performance that is achievable with any practical algorithm.   We derive performance  in terms of  probability of error of the ML decoder   when  sampling is performed via an arbitrary  linear sampling operator.  Based on an upper bound on the probability of error, we derive the minimum number of samples required for asymptotically  reliable recovery  of subspaces in terms of a SNR measure, the dimension of each subspace  in the union and a term which quantifies the dependence or overlap   among the subspaces.
In the special case where sampling is performed via random projections and the subspaces in the union have a specific structure such that each subspace is a sum of some other  $k_0$ subspaces, we obtain a more explicit expression for the minimum number of measurements. This number depends on the number of underlying subspaces, the dimension of  each subspace, and the minimum nonzero  block SNR (defined in Section \ref{random_block}). We note that the conventional sparsity model is a special case of this structure.

The asymptotic probability of error  of the ML decoder for sparse support recovery   in  the presence of noise for the standard sparsity model  was first investigated in  \cite{wain2} followed by several other authors  \cite{Fletcher1,wang5,thakshilaj4}. In \cite{wain2},  sufficient conditions were derived on  the number of noisy compressive measurements needed to achieve  a vanishing probability of error asymptotically for sparsity pattern recovery while    necessary conditions were considered in \cite{Fletcher1}.  The analyses in both \cite{wain2} and \cite{Fletcher1} are  based on the assumption that the sampling operator is random. Here, we follow a similar path  assuming  the union of subspaces model. However, there are some key differences  between  our derivations  and  that in \cite{wain2}. First, we treat arbitrary (not necessarily random)  sampling operators  and assume a  general union of subspaces model as opposed to  the standard sparsity model.  Further, the  results in \cite{wain2}  were derived based on weak bounds on the probability of error, thus there is a gap between those results and the number of measurements required for the exact probability of error to vanish asymptotically at finite SNR. We consider tighter bounds on the probability error leading to tighter results.

The rest of the paper is organized as follows. In Section  \ref{formulation}, the problem of subspace recovery   from  a union of subspace model is introduced. In Section \ref{ML_decoder}, performance limits  with ML  decoder for subspace recovery   in terms of the probability of error are derived with a given linear sampling operator considering a general union of subspaces model. Conditions under which asymptotically  reliable subspace recovery  in the presence of noise is guaranteed   are obtained based on the derived upper bound. The results are extended in Section \ref{structured} to the setting where structured properties of the subspaces in the union are available.  We also derive sufficient conditions for subspace recovery  when sampling is performed via random projections. In Section \ref{discussion}, we compare our results with some existing results in the literature.  Practical algorithms  to recover  subspaces in the union of subspace model and  numerical results to validate the theoretical claims are presented in Section  \ref{numerical}.

Throughout the paper, we use the following notation. Arbitrary vectors in a Hilbert space  $\mathcal H$, are denoted by lower case letters, e.g., $x$. Calligraphic letters, e.g., $\mathcal S$, are used to represent subspaces in $\mathcal H$.  Vectors in $\mathbb R^N$ are written in boldface lower case letters, e.g. $\mathbf x$. Scalars (in $\mathbb R$) are also denoted by lower case letters, e.g., $x$, when there is no confusion.  Matrices are written in boldface upper case letters, e.g., $\mathbf A$. Linear operators and  a set of basis vectors for a given subspace $\mathcal S$ are denoted by  upper case letters, e.g., A. The notation  $\mathbf x \sim \mathcal N (\boldsymbol\mu, \Sigma)$ means  that the random vector $\mathbf x$ is distributed as multivariate  Gaussian with mean $\boldsymbol\mu$ and the covariance matrix $\Sigma$; $x\sim \mathcal X_m^2(\lambda)$ denotes that the random variable $x$ is distributed as Chi squared with $m$ degrees of freedom and  non centrality parameter $\lambda$. (The central Chi squared distribution is denoted by $\mathcal X_m^2$).    By $\mathbf 0$, we denote  a vector with appropriate dimension in which all elements are zeros, and $\mathbf I_k$ is  the identity matrix of size $k$.  The conjugate transpose of a matrix $\mathbf A$ is denoted by $\mathbf A^*$. Finally,  $||.||_2$ denotes the $l_2$ norm and $|.|$ is used for  both the cardinality (of a set) and the absolute value (of a scalar).  Special functions used in the paper are:  Gaussian $Q$-function:
\begin{eqnarray}
Q(x) = \frac{1}{\sqrt{2\pi}} \int_x^{\infty} e^{-\frac{t^2}{2}} dt\label{q_func}
\end{eqnarray}
Gamma function:
\begin{eqnarray}
\Gamma(x) = \int_0^{\infty} t^{x-1} e^{-t} dt\label{gamma_func}
\end{eqnarray} and
 modified Bessel function with real arguments:
 \begin{eqnarray}
 K_\nu(x) = \int_0^{\infty} e^{-x \mathrm{cosh} t} \mathrm{cosh} (\nu t) dt.\label{bessel_func}
 \end{eqnarray}

\section{Problem Formulation}\label{formulation}

\subsection{Union of subspaces}
As discussed in \cite{Blumensath1,Lu1,Eldar1,Baraniuk4}, there are many practical scenarios where the signals of interest lie in a union of subspaces.

  \begin{definition}{Union of subspaces: }
  A signal $x\in \mathcal H$ lies in a   union of subspaces if $x \in \mathcal X$ where $\mathcal X$ is defined as
  \begin{eqnarray}
   \mathcal X = \underset{i} {\bigcup} ~ \mathcal S_i\label{union}
    \end{eqnarray}
    and $\mathcal S_i$'s are  subspaces of $\mathcal H$ which are assumed to be finite dimensional.  A signal $x \in \mathcal X$ if and only if there exists $i_0$ such that $x \in \mathcal S_{i_0}$.
  \end{definition}

  Let $V_i = \{ v_{im}\}_{m=0}^{k-1}$ be a basis for the finite dimensional subspace $\mathcal S_i$ where $k$ is the dimension of  $\mathcal S_i$ (it is noted that while we assume all subspaces to  have the same dimension, the analysis can be easily extended for the case where different subspaces have different dimensions). Then each  $x\in  \mathcal S_i$ can be  expressed in terms of a  basis expansion
  \begin{eqnarray*}
  x = \sum_{m=0}^{k-1} c_{i}(m) v_{im}
  \end{eqnarray*}
  where $c_{i}(m)$'s  for $m=0,1,\cdots k -1$ are the coefficients corresponding to  the basis $V_i$. We assume that the  subspaces are distinct  (i.e. there are no  subspaces such  that $\mathcal S_i \subseteq \mathcal S_j$ for $i\neq j$ in the union (\ref{union})) and each subspace $\mathcal S_i$ is uniquely determined by  the basis $V_i$. We denote by   $T < \infty$ the  number of  subspaces in the union $\mathcal X$.

\subsection{Structured  union of subspaces leading to block sparsity}\label{structured_block}
There  are certain scenarios in which the signals can be assumed to lie in more structured union of subspaces as considered in \cite{Eldar1,Haim2,Duarte2}.   Suppose that  each subspace in the union (\ref{union}) can be  represented as a sum of $k_0$ (out of $L$) disjoint subspaces \cite{Eldar1,Haim2}. More specifically,
\begin{eqnarray}
\mathcal S_i = \underset{j \in \Sigma_{k_0}}{\oplus} \mathcal V_j\label{subspace_addition}
\end{eqnarray}
where $\{\mathcal V_j\}_{j=0}^{L-1}$'s are disjoint subspaces, and  $\Sigma_{k_0}$ contains $k_0$ indices from $\{0,1,\cdots, L-1\}$. Let $d_j = \dim (\mathcal V_j)$ and $ N = \sum_{j=0}^{L-1} d_j$. Then there are $T={L\choose k_0}$  subspaces in the union. Under this formulation, the dimension of each subspace in  (\ref{union}) is $k= \underset{j \in \Sigma_{k_0}} {\sum} d_j$. In the special case where $d_j= d$ for all $j$,  $k=k_0d$.

Now taking  $V_j$ as a basis for $\mathcal V_j$,
a signal in the union can be written as
\begin{eqnarray}
x = \underset{j\in \Sigma_{k_0}}{\sum} V_j \mathbf c_j \label{subspace_1}
\end{eqnarray}
where $\mathbf c_j = [c_j(0), \cdots, c_j(d_j-1)]^T \in \mathbb{R}^{d_j}$ is a $d_j \times 1$ coefficient vector corresponding to the basis $V_j$. It is worth mentioning  that we use the same notation $V_j$ to denote a basis of the subspace $\mathcal S_j$ in (\ref{union}) for $j=0,1,\cdots, T-1$  (when discussing the general union of subspaces model) and  also to denote a basis of the subspace $\mathcal V_j$ in (\ref{subspace_addition}) for $j=0,1,\cdots, L-1$  (when discussing the structured union of subspace model). Let $V$ be a matrix constructed by concatenating  $V_i$'s column wise, such that $V=[V_0 | V_1 |\cdots | V_{L-1}]$ and $\mathbf c$ be a $N\times 1$ vector with $\mathbf c = [\mathbf c^T_{0}| \cdots |\mathbf c^T_{L-1}]^T$. As defined in \cite{Eldar1}, the vector $\mathbf c \in \mathbb R^N$ is called block $k$-sparse over $\mathcal I = \{d_0, d_1, \cdots, d_{L-1}\}$ if all the elements in $\mathbf c_i$ are zeros for all but $k_0$ indices where $N = \sum_{j=0}^{L-1} d_j$. In this paper, we assume $d_j =d $ for all $j$ so that $N=Ld$.

  The standard sparsity model used in the CS literature is  a special case of this structured union of subspaces  model when $d=1$. In the standard CS framework, $x= \mathbf x$ is a length-$N$ signal vector    which is  $k$-sparse in an $N$-dimensional  orthonormal basis $\mathbf V$  so that  $\mathbf x$ can be represented as $\mathbf x = \mathbf V  \mathbf c$ with $\mathbf c$ having only $k\ll N$ significant coefficients. This fits our framework when  $d=1$ and $V_j$ is chosen  as the $j$-th column vector of the orthonormal basis  $\mathbf V$ for $j=0,1,\cdots, N-1$.   In this case, we have $L=N$ and  there are  $T = {N\choose k}$ subspaces in the union.


  \subsection{Observation model: Linear sampling}

  Consider a sampling operator via a bounded linear mapping of a signal $x$ that lies in an ambient  Hilbert space $\mathcal H$. Let the linear sampling operator $A$ be specified by a set of unique sampling vectors  $\{a_m\}_{m=0}^{M-1}$.  With these notations, noisy samples are given by,
\begin{eqnarray}
\mathbf y = A x  + \mathbf{ w}\label{obs_1}
\end{eqnarray}
  where $\mathbf y$ is the $M\times 1$ measurement vector, and the $m$-th element of the vector $Ax$ is given by, $(Ax)_m = \langle x, a_m \rangle$ for $m=0,1,\cdots, M-1$  where $\langle., .\rangle$ denotes the inner product. The noise vector $\mathbf w$ is assumed to be Gaussian with mean $\mathbf 0$ and  covariance matrix $\sigma_w^2 \mathbf I_M$.

  When  $x\in  \mathcal S_i$ for some $i$ in the model (\ref{union}), the vector of samples  can be equivalently  represented  in the  form of a matrix vector multiplication,
  \begin{eqnarray}
  \mathbf y = \mathbf B_i \mathbf c_i + \mathbf w\label{obs_2}
  \end{eqnarray}
  where
  \begin{eqnarray*}
  \mathbf B_i = A V_i = \left(
  \begin{array}{cccccc}
  \inner{a_0}{v_{i0}} & \inner{a_0}{v_{i1}} & \cdots & \inner{a_0}{v_{i(k-1)}}\\
   \inner{a_1}{v_{i0}} & \inner{a_1}{v_{i1}} & \cdots & \inner{a_1}{v_{i(k-1)}}\\
   \vdots  &\vdots & \vdots  & \vdots\\
   \inner{a_{M-1}}{v_{i0}} & \inner{a_{M-1}}{v_{i1}} & \cdots & \inner{a_{M-1}}{v_{i(k-1)}}
    \end{array}\right)
  \end{eqnarray*}
  and $\mathbf c_i = [c_{i}(0) ~c_{i}(1) ~ \cdots c_{i}(k-1)]^T$ is the coefficient vector with respect to the basis $V_i$. Further, let $\mathbf b_{im}$ denote the $m$-th column vector of the matrix $\mathbf B_i$ for $m=0,1,\cdots, k-1$ and $i = 0,1,\cdots, T-1$. We assume that the linear sampling operator $A$ is a one-to-one mapping between $\mathcal X$ and $A\mathcal X$.  Since  $\{v_{i0}, \cdots, v_{i(k-1)}\}$ is a set of  linearly independent basis vectors,  then $\{\mathbf b_{i0}, \cdots, \mathbf b_{i(k-1)}\}$ are also linearly independent for each $i=0,1,\cdots, T-1$. It is worth noting that, while this one-to-one condition ensures uniqueness,  stronger conditions are required to recover $x$ in a stable manner as discussed in  \cite{Lu1}.


\subsection{Subspace recovery  from  the union of subspaces model}
As discussed in the Introduction, there are applications where it is sufficient to recover  the subspace in which   the signal of interest lies from the union of subspaces model (\ref{union}) instead of complete signal recovery. Moreover,  if there is a procedure to correctly identify the subspace  with vanishing  probability of error, then   the signal $x$ can be reconstructed with a small  $l_2$ norm error using standard techniques.  However, the other way would not be always true, i.e., if an algorithm developed for complete signal recovery is used for subspace recovery, it may  not give an equivalent performance guarantee. This is because, even if such an estimate of the signal may be close to the true signal with respect to the considered performance metric (e.g., $l_2$ norm error), the subspace in which the estimated signal lies  may be different from the true subspace.  This can happen especially when the SNR is not sufficiently large. Thus, investigating the problem of subspace recovery is important and  is the main focus of this paper.

The problem of subspace recovery  is to identify the subspace in which the signal $x$ lies. The estimated subspace, $\hat {\mathcal S}$,  via any  recovery scheme  can be expressed in the following form:
\begin{eqnarray}
\hat {\mathcal S} = \zeta (\mathbf y)\label{estimate_1}
\end{eqnarray}
where $\zeta(\cdot)$ is a mapping from the observation vector $\mathbf y$ to an estimated subspace $\hat{\mathcal S} \in \{\mathcal S_0, \cdots, \mathcal S_{T-1}\}$.
The performance metric used to evaluate the quality of the estimate (\ref{estimate_1}) is taken as the average probability of error defined as
\begin{eqnarray}
P_e = \underset{\mathcal S}{\sum} Pr(\zeta(\mathbf y) \neq \mathcal S | \mathcal S) Pr (\mathcal S)
\end{eqnarray}
for a given recovery scheme   $\zeta (\mathbf y)$. We say that the mapping $\zeta (\mathbf y)$ is capable of providing \emph{asymptotically reliable} subspace recovery   if $P_e\rightarrow 0$ as $M\rightarrow \infty$.
In this paper, we consider  subspace recovery  via the ML estimation.
Our goal is to address the following issues.

\begin{itemize}
\item Performance of the  ML estimation scheme   in terms of  the probability of error in recovering  the subspaces from the union of subspaces model (\ref{union}) in the presence of noise. We are also  interested in  conditions under which  asymptotically  reliable  subspace  recovery in  the union  is guaranteed  with a given sampling operator.
    \item How much gain in terms of the  number of samples required for subspace recovery  can be achieved if further information on structures is available for the subspaces in (\ref{union}) compared to the case when  no additional structured information is available   (i.e. compared to the  standard sparsity model used in CS).
        \item Illustration of the performance gap between the ML estimation  and  computationally tractable  algorithms for subspace recovery  from the union of subspaces model at finite SNR.
        \end{itemize}

The main results of the paper can be summarized as follows. With the general union of subspaces model as defined in (\ref{union}), and for a given sampling operator, the minimum number of samples required for asymptotically  reliable recovery  of subspaces in the presence of noise is
 \begin{eqnarray}
 M > k+ \frac{\eta_3}{f(SNR)} \log(\bar T_0)\label{M_general}
  \end{eqnarray}where $k$ is the dimension of each subspace, $f(SNR)$ is a  measure  of the minimum SNR of the sampled signal projected onto the null space of any subspace in the union,  $\bar T_0$ is a measure of the number of subspaces in the union with  maximum dependence where $\bar T_0 \leq T$ (formal definitions of all these terms are given  in Section \ref{ML_decoder}), and $\eta_3$ is a constant. We simplify (\ref{M_general}) for  the special case where each subspace in the union (\ref{union}) can be expressed as a sum of  $k_0$ subspaces out of $L$  where each such subspace is $d$-dimensional such   that $k=k_0d$. Then, the problem of subspace recovery  reduces to the problem of block sparsity pattern recovery. Further, assuming that the sampling operator is represented by random projections, the number of samples required for asymptotically  reliable block sparsity pattern recovery is given by \begin{eqnarray}
   M > k + \frac{\eta_4}{BSNR_{\min}} \log(L-k_0)\label{M_block}
    \end{eqnarray}
    where $BSNR_{\min}$ is the minimum nonzero block SNR and  $\eta_4$ is a constant. When $d=1$ and $L=N$ where $N$ is the signal dimension, the block sparsity model reduces to the standard sparsity model. Then,  our result shows that
      \begin{eqnarray}
     M >  k + \frac{\eta_2}{ CSNR_{\min}}\log(N-k) \label{M_standard}
      \end{eqnarray}
      measurements are required for reliable sparsity pattern recovery where $CSNR_{\min}(\leq \frac{BSNR_{\min}}{d})$ is the minimum component SNR of the signal.  Thus,  from (\ref{M_block}) and (\ref{M_standard}), we observe that the number of measurements  required for asymptotically reliable  subspace recovery beyond the sparsity index (i.e., $M-k$) reduces  approximately $d$ times  with a block sparsity model (so that $k=k_0 d$) compared to the standard  sparsity model.
     A detailed comparison between  our results and  existing results in the literature  is  given  in  Section \ref{discussion}.


\section{Subspace Recovery    With General Unions}\label{ML_decoder}
  The problem of finding the true subspace from the union (\ref{union}) based on  the observation model (\ref{obs_2}) via the ML estimation  becomes finding the index $\hat i$ such that,
\begin{eqnarray*}
\hat i = \underset{i = 0,\cdots, T-1}{\arg \max} ~ p(\mathbf y |  \mathbf B_i).
\end{eqnarray*}
When  $x\in \mathcal S_i$ in (\ref{union}) for some $i$,  and using  the observation model (\ref{obs_2}), we have $
p(\mathbf y | \mathbf B_i) = \mathcal N (\mathbf B_i \mathbf c_i, \sigma_w^2 \mathbf I_M).
$
The signal $x$ is assumed to be deterministic but unknown. Thus,  when $x\in \mathcal S_i$, the  coefficient vector  $\mathbf c_i$ with respect to a given  basis $\mathbf B_i$ is unknown. Assuming that each $\mathbf B_i$ has rank $k$ for $i=0,\cdots, T-1$, the ML estimate of  $\mathbf c_i$  such that $p(\mathbf y | \mathbf B_i) $ is maximized can be found as,
$
\hat{\mathbf c_i}  = (\mathbf B_i^* \mathbf B_i)^{-1} \mathbf B_i ^* \mathbf y
$. This results  in
\begin{eqnarray*}
\log (\underset{\mathbf c_i}{\max} ~ p(\mathbf y | \mathbf B_i)) &=& \log \left(\frac{1}{(2\pi \sigma_w^2)^{M/2}}\right) - \frac{1}{2\sigma_w^2} || \mathbf y - \mathbf P_i \mathbf y||_2^2 \nonumber\\
&=&  \log \left(\frac{1}{(2\pi \sigma_w^2)^{M/2}}\right) - \frac{1}{2\sigma_w^2} ||\mathbf P^\bot_i \mathbf y||_2^2
\end{eqnarray*}
where $ \mathbf P_i  = \mathbf B_i(\mathbf B_i^* \mathbf B_i)^{-1} \mathbf B_i ^*  $ is the orthogonal projector onto the span of  $\{\mathbf b_{im}\}_{m=0}^{k-1}$ and $\mathbf P^\bot_i = \mathbf I - \mathbf P_i$.
Thus, the estimated index of the subspace by the ML estimation  is,
\begin{eqnarray}
\hat i = \underset{i=0,\cdots, T-1}{\arg\min } ||\mathbf P^\bot_i \mathbf y||_2^2.
\end{eqnarray}

The   probability of error of the ML estimation  is given by,
\begin{eqnarray}
P_e &=& Pr(\mathbf B_{estimated}\neq \mathbf B_{true}) = \underset{i} {\sum} Pr(\hat i \neq i | \mathbf  B_i) Pr(\mathbf B_i)\nonumber\\
& \leq & \underset{i} {\sum} \underset{j\neq i} {\sum}  Pr( \hat i= i | \mathbf B= \mathbf B_j) Pr(\mathbf B=\mathbf B_j)\label{PE_ML_Union_bound}
\end{eqnarray}
where $Pr( \hat i = i | \mathbf B= \mathbf B_j)$ is the probability of selecting $\mathcal S_i$ when the true subspace is $\mathcal S_j$. Since the   ML estimation  decides   the subspace $\mathcal S_i$ over $\mathcal S_j$  when $||\mathbf P^\bot_i \mathbf y||_2^2 - ||\mathbf P^\bot_j \mathbf y||_2^2 < 0$, $Pr( \hat i = i | \mathbf B= \mathbf B_j)$ is given by
\begin{eqnarray*}
Pr( \hat i = i | \mathbf B= \mathbf B_j)  = Pr (||\mathbf P^\bot_i \mathbf y||_2^2 - ||\mathbf P^\bot_j \mathbf y||_2^2 < 0 )
\end{eqnarray*}
for $i\neq j$.

%

Let $\Delta_{ij}(\mathbf y) = ||\mathbf P^\bot_i \mathbf y||_2^2 - ||\mathbf P^\bot_j \mathbf y||_2^2$ for $i\neq j$. When the true subspace is $\mathcal S_j$ so that $Ax = \mathbf B_j \mathbf c_j$, we have $||\mathbf P^\bot_j \mathbf y||_2^2 = ||\mathbf P^\bot_j \mathbf w||_2^2$ and
\begin{eqnarray}
\mathbf P^\bot_i \mathbf y &=& \mathbf P^\bot_i A x  + \mathbf P^\bot_i \mathbf w\nonumber\\
 &=&\mathbf P^\bot_i \mathbf B_j \mathbf c_j + \mathbf P^\bot_i \mathbf w= \mathbf P^\bot_i \mathbf B_{j\setminus i} \mathbf c_{j\setminus i} + \mathbf P^\bot_i \mathbf w
\end{eqnarray}
where $\mathbf B_{j\setminus i} \mathbf c_{j\setminus i} = \underset{\mathbf b_{jm}\notin \mathcal R(\mathbf B_i)}{\sum} \mathbf b_{jm} c_j(m)$ and  $\mathcal R(\mathbf A)$ denotes the range space of  the matrix $\mathbf A$. More specifically, the  $M \times l$ matrix $\mathbf B_{j\setminus i}$  contains  the columns of $\mathbf B_j$ which are not in the range space of the matrix $\mathbf B_i$ where $l$ is  the number of columns in $\mathbf B_{j\setminus i}$. The  $l\times 1$ vector $\mathbf c_{j\setminus i}$ contains the elements of $\mathbf c_j$ corresponding to the column vectors in $\mathbf B_{j\setminus i}$.

We conclude that,  the decision statistic for selecting $\mathcal S_i$ over $\mathcal S_j$ is given by
$\Delta_{ij}(\mathbf y) = ||\mathbf P^\bot_i(\mathbf B_{j \setminus i} \mathbf c_{j\setminus i} + \mathbf w)||_2^2 - || \mathbf P^\bot_j \mathbf w||_2^2$ and $Pr(\Delta_{ij}(\mathbf y) < 0)$ $ = Pr\left(\frac{||\mathbf P^\bot_i(\mathbf B_{j \setminus i} \mathbf c_{j\setminus i} + \mathbf w)||_2^2}{||\mathbf  P^\bot_j \mathbf w||_2^2} < 1\right)$.   When  $\mathbf B_j$ is given, the random variable $g_1 = ||\mathbf P^\bot_i(\mathbf B_{j \setminus i} \mathbf c_{j\setminus i} + \mathbf w)||_2^2 / \sigma_w^2$ is a non-central Chi squared random variable with $M-k$ degrees of freedom and  non-centrality parameter $||\mathbf P^\bot_i(\mathbf B_{j \setminus i} \mathbf c_{j\setminus i} )||_2^2/\sigma_w^2$. The random variable  $g_2 = || \mathbf P^\bot_j \mathbf w||_2^2 / \sigma_w^2$ is a (central) Chi-squared random variable with $M- k$ degrees of freedom.
The two random variables $g_1$ and $g_2$ are, in general,  correlated and  the computation of the exact value of $Pr\left(\Delta_{ij}(\mathbf y) < 0\right)$ is difficult. In the following we find an upper bound for the quantity $Pr\left(\Delta_{ij}(\mathbf y) < 0\right)$ following  techniques similar to those  proposed in \cite{wain2}.

\subsection{Upper bound on $Pr\left(\Delta_{ij}(\mathbf y) < 0\right)$}
For  clarity, we  introduce  the following notation.  Let $\mathcal W_{j\setminus i}$ be  the set consisting of column indices of $\mathbf B_j$  such that $\mathbf b_{jm} \notin \mathcal R(\mathbf B_i)$ for $m=0,1,\cdots k-1$ and  $i\neq j$ ($i,j, = 0,1,\cdots, T-1$). We then have that $ |\mathcal W_{j\setminus i}| =l$ where  $l$ can take values from $1,2,\cdots, k$.

\begin{lemma}\label{lammaU_1}
Assume that the  sampling operator $A$ is known. For any  given signal $x\in \mathcal S_j$ so that $A x = \mathbf B_j \mathbf c_j$,  the probability of error  in selecting the subspace $\mathcal S_i$ over $\mathcal S_j$, $Pr(\Delta_{ij}(\mathbf y) < 0)$,   is upper bounded  by,
\begin{eqnarray}
Pr(\Delta_{ij}(\mathbf y) < 0) \leq Q \left(\frac{1}{2}(1 - 2\eta_0 )\sqrt{\lambda_{j\setminus i}} \right) +  \Psi\left(l, \lambda_{j\setminus i}\right) \label{p_delta_y}
\end{eqnarray}
where $\lambda_{j\setminus i} =  \frac{1}{\sigma_w^2} || {\mathbf P}_i^\bot  {\mathbf B}_{j\setminus i}{\mathbf c}_{j\setminus i}||_2^2$, $\Psi\left(l, \lambda_{j\setminus i}\right) =  \frac{\sqrt 2}{2^{l}\Gamma(l/2)}(\eta_0\lambda_{j\setminus i} )^{l/2-1/2} K_{l/2-1/2}\left(\frac{\eta_0\lambda_{j\setminus i}}{2}\right)$,    $Q(x)$ is the Gaussian $Q$ function (\ref{q_func}),   $\Gamma(x)$ is the Gamma function (\ref{gamma_func}),   $K_{\nu}(x)$ is the modified Bessel function (\ref{bessel_func}), and  $0 < \eta_0 < \frac{1}{2} $.
\end{lemma}

\begin{proof}
See Appendix A.
\end{proof}
\begin{theorem}\label{theorem1}
Assuming that the true subspace  is chosen uniformly at random   from $T$ subspaces in  the union (\ref{union}),  the average probability of error of the ML estimation  for subspace recovery   is upper bounded by,
\begin{eqnarray}
P_e \leq \frac{1}{T } \sum_{i=0}^{T-1}\sum_{j=0}^{T-1} Q \left(\frac{1}{2}(1 - 2\eta_0 )\sqrt{\lambda_{j\setminus i}} \right) +  \Psi\left(l, \lambda_{j\setminus i}\right) \label{Pe_bound0}
\end{eqnarray}
where $\lambda_{j\setminus i} $,  $\eta_0 $, $Q$, $\Psi$ are as defined in Lemma \ref{lammaU_1}.
\end{theorem}

\begin{proof}
The proof follows from Lemma \ref{lammaU_1} and (\ref{PE_ML_Union_bound}).
\end{proof}

In general,  the subspaces $\mathcal S_i$ and $\mathcal S_j$ can overlap; i.e. there can be elements in $\mathcal S_j$ which are also in $\mathcal S_i$. However, one subspace can not lie in another subspace entirely; i.e. all subspaces $\mathcal S_i$'s are distinct for $i=0,1,\cdots, T-1$.   As defined before,  $\mathcal W_{j\setminus i}$ contains  the column  indices of $\mathbf B_j$ which are not in $\mathcal R(\mathbf B_i)$ and $|\mathcal W_{j\setminus i}|=l$ for any $i\neq j$ where $l$  takes values from $1,2,\cdots,k$.  As $l$ increases, the overlap   of the two subspaces decreases resulting in more separable subspaces. In the special case where  $\mathcal S_j$ and $\mathcal S_i$ do not intersect  at all, we have $l=k$.  Thus, $l$ can be considered as a measure of overlap  between any two subspaces $\mathcal S_j$ and $\mathcal S_i$ for $i\neq j$ in the union (\ref{union}).    For given $l$,  the probability $Pr(\Delta_{ij}(\mathbf y) < 0)$ in (\ref{p_delta_y})  monotonically decreases as $\lambda_{j\setminus i}$, defined in Lemma \ref{lammaU_1},  increases. This implies that when  $\lambda_{j\setminus i}$ is large, the probability of selecting  $\mathcal S_i$ as the true subspace (given that the true subspace is $\mathcal S_j$) decreases. In other words, $\lambda_{j\setminus i}$, is used to characterize the error in selecting the subspace $\mathcal S_i$ over $\mathcal S_j$ for $i\neq j$ (or how distinguishable the subspace $\mathcal S_i$ is with respect to  $\mathcal S_j$) when the true subspace is $\mathcal S_j$.  It is, therefore,  of interest to further investigate the quantity $\lambda_{j\setminus i}$.

\subsection{Evaluation of  $\lambda_{j\setminus i}$}\label{eval_lamda}
For any given signal $x\in \mathcal S_j$, as defined in Lemma \ref{lammaU_1}, $\lambda_{j\setminus i}$ is given by,
\begin{eqnarray*}
\lambda_{j\setminus i} =  \frac{1}{\sigma_w^2} || {\mathbf P}_i^\bot  Ax ||_2^2 = \frac{1}{\sigma_w^2} || {\mathbf P}_i^\bot  {\mathbf B}_{j\setminus i}{\mathbf c}_{j\setminus i}||_2^2.
\end{eqnarray*}
When  the true subspace is assumed to be $\mathcal S_j$, the quantity $|| {\mathbf P}_i^\bot  {\mathbf B}_{j\setminus i}{\mathbf c}_{j\setminus i}||_2^2 ~(=|| {\mathbf P}_i^\bot  {\mathbf B}_j{\mathbf c}_{j}||_2^2 = || {\mathbf P}_i^\bot  Ax||_2^2)$ denotes the energy of the sampled signal $Ax$ projected onto the null space of $\mathbf B_i$; i.e., the energy of the sampled signal which is unaccounted for by $\mathcal S_i$ for $i\neq j$. Therefore, when $|| {\mathbf P}_i^\bot  {\mathbf B}_{j\setminus i}{\mathbf c}_{j\setminus i}||_2^2$ is large, the probability that the subspace $\mathcal S_i$ is  selected as the true subspace becomes small.
Further, if $\mathcal S_j \subseteq \mathcal S_i$ for any $\mathcal S_i$, we have $|| {\mathbf P}_i^\bot  {\mathbf B}_{j\setminus i}{\mathbf c}_{j\setminus i}||_2^2 = 0$. However, this cannot happen based on our assumption that there is no subspace in the union which completely overlaps another. Thus, $\lambda_{j\setminus i} > 0$.

Let the eigendecomposition of ${\mathbf P}_i^\bot$ be ${\mathbf P}_i^\bot = \mathbf Q_i \boldsymbol\Lambda_i \mathbf Q_i^T$ where $\mathbf Q_i$ is a unitary matrix consisting of eigenvectors of ${\mathbf P}_i^\bot$ and   $\boldsymbol\Lambda_i$ is a diagonal matrix in which the diagonal elements represent  eigenvalues of ${\mathbf P}_i^\bot$ which are $M-k$  ones and $k$ zeros. Then, for given $l$,
\begin{eqnarray}
\lambda_{j\setminus i} = \frac{1}{\sigma_w^2} || {\mathbf P}_i^\bot  {\mathbf B}_{j\setminus i}{\mathbf c}_{j\setminus i}||_2^2 = \underset{m\in \mathcal Q_i}{\sum} \alpha_{m,i}^2(l)  \geq (M-k) \alpha_{\min,l}^2\label{alpha_min}
\end{eqnarray}
where $\alpha_{m,i}(l) = \frac{1}{\sigma_w}  \langle  \mathbf q_{m,i} ,  {\mathbf B}_{j\setminus i}{\mathbf c}_{j\setminus i}\rangle$ for given $l$,  $\mathbf q_{m,i}$ is the $m$-th eigenvector of ${\mathbf P}_i^\bot$, $\mathcal Q_i$ is the set containing indices corresponding to nonzero eigenvalues where $|\mathcal Q_i| = M-k$ and $\alpha_{\min,l} =\underset{i; i\neq j}{\min} |\alpha_{m,i}(l)|$.


Note that $ (M-k) \alpha_{\min,l}^2$ is a measure of  the minimum SNR of the sampled signal, $Ax$, projected onto the null space of any subspace $\mathcal S_i$ for $i\neq j$,  $i=0,1,\cdots,T-1$ such that $|\mathcal W_{j\setminus i}|=l$  given that the true subspace in which  the signal lies is $\mathcal S_j$.

For a given subspace  $\mathcal S_j$, define $T_j(l)$ to be the number of subspaces $\mathcal S_i$ such that $|\mathcal W_{j\setminus i}| = l$.  With these notations, the probability of error in (\ref{Pe_bound0}) can be further upper bounded by,
\begin{eqnarray}
P_e &\leq& \frac{ 1}{T} \sum_{j=0}^{T-1} \sum_{l=1}^{k} T_j(l) \left(Q \left(\frac{1}{2}(1 - 2\eta_0 )\sqrt{(M-k) \alpha_{\min,l}^2} \right)\right. \nonumber\\
&+& \left. \Psi\left(l, {(M-k) \alpha_{\min,l}^2}\right)\right) \label{pe_ML1}
\end{eqnarray}
where $\Psi\left(l, {(M-k) \alpha_{\min,l}^2}\right) = \frac{\sqrt 2}{2^{l}\Gamma(l/2)}(\eta_0(M-k) \alpha_{\min,l}^2  )^{l/2-1/2} K_{l/2-1/2}(\eta_0(M-k) \alpha_{\min,l}^2 /2)$.
To obtain  (\ref{pe_ML1}) we used the facts that $Q(x)$ is  monotonically non increasing in $x$ and $\Psi(s,x)$ is monotonically non increasing in  $x$ for given $s$ when $x>0$. The quantity $T_j(l)$ is a measure of the overlap  between  $\mathcal S_j$ and any subspace $\mathcal S_i$ for $i\neq j, i=0,1,\cdots, T-1$. To compute $T_j(l)$ explicitly,  the specific structures of  the subspaces should be known. For example, in the standard sparsity model used in CS in which the union in (\ref{union}) consists of $T={N\choose k}$ subspaces from an orthonormal basis $\mathbf V$ of dimension $N$, there are ${k\choose l}{{N-k}\choose l}$ number of sets such that $|\mathcal W_{j\setminus i}|=l$, thus $T_j(l) = {k\choose l}{{N-k}\choose l}$. In that particular case,  $T_j(l)$ is the same for all $j=0,1,\cdots,T-1$.     To  further upper bound (\ref{pe_ML1}), we let
\begin{eqnarray}
T_0(l) = \underset{j=0,1,\cdots,T-1} {\max}T_j(l). \label{T_0}
\end{eqnarray} Then,
\begin{eqnarray}
P_e &\leq&   \sum_{l=1}^{k} T_0(l) \left(Q \left(\frac{1}{2}(1 - 2\eta_0 )\sqrt{(M-k) \alpha_{\min,l}^2} \right) \right. \nonumber\\
&+&\left. \Psi\left(l, {(M-k) \alpha_{\min,l}^2}\right)\right)\label{pe_ML}.
\end{eqnarray}

\begin{theorem}\label{lemma_M_conditions}
Let $\alpha_{\min,l}^2$ and $T_0(l) $ be as defined in (\ref{alpha_min}) and (\ref{T_0}), respectively.   Suppose that  sampling is performed via a  sampling operator $A$. Then $P_e $ in (\ref{pe_ML}) vanishes asymptotically (i.e., $\underset{(M-k)\rightarrow \infty}{\lim} P_e \rightarrow 0$) if the following condition is satisfied:
\begin{eqnarray*}
M > k + \max\{M_1, M_2\}
\end{eqnarray*}
 where
 \begin{eqnarray}
M_1&=& \underset{l=1,\cdots,k}{\max} \left\{f_1(l) \right.\nonumber\\
&=&\left.\frac{8}{(1-2\eta_0)^2\alpha_{\min,l}^2}
  \left\{\log(T_0(l))   + \log(1/2)\right\}\right\}
 \end{eqnarray}
 \begin{eqnarray}
  M_2 &=& \underset{l=1,\cdots,k}{\max} \left\{f_2(l)\right. \nonumber\\
  &=& \left. \frac{2(k/2 + r_0 -1)}{r_0\eta_0   \alpha_{\min,l}^2}
 \left\{\log(T_0(l)) + \log\left(\frac{2b_0}{\sqrt\pi}\right)\right\} \right\}
 \end{eqnarray} with   $0<\eta_0 < 1/2 $, $b_0 = \frac{\sqrt{2\pi}}{4}$ and $r_0 >0$.
\end{theorem}

\begin{proof}
See Appendix B.
\end{proof}
Let $l_i \in\{1,\cdots,k\}$   be the value of $l$ which maximizes $f_i(l)$  as defined in Theorem \ref{lemma_M_conditions} for $i=1,2$.
For  $M_2$, it can be verified that we can  find constants $\eta_0$ and $r_0$ in the defined regimes such that $\frac{8}{(1-2\eta_0)^2} > \frac{2(k/2+r_0 - 1)}{r_0 \eta_0}$ if $k$ is fairly small. Then  the dominant factor of $M_1$ and $M_2$  can be written in the form of $\frac{ \eta_3}{\bar \alpha_{\min}^2} \log(\bar T_0)$  where $\bar{\alpha}_{\min}^2$ and $\bar T_0$ are the  corresponding values of $\alpha_{\min,l}^2$ and $T_0(l)$ when  $l=l_0$ for $l_0 \in\{l_1,l_2\}$ and  $ \eta_3$ is an appropriate constant.
Since, most of the scenarios we are interested in are for the case where $k$ is sufficiently small, we get the minimum number of samples required for reliable subspace recovery  as
\begin{eqnarray}
M \geq k + \frac{ \eta_3}{\bar \alpha_{\min}^2} \log(\bar T_0).
  \end{eqnarray}
It is further noted that $T_0(l)\leq T$ for all $l$ and thus $\bar T_0 \leq T$ where $T$ is the total number of subspaces in the union (\ref{union}).

\subsection{Random sampling}\label{random_sampling}
Next, we consider  the special case where the sampling operator   is a $M\times N$ matrix in which the elements are realizations of a random variable (e.g. Gaussian). Then we have $\mathbf B_i = \mathbf A \mathbf V_i$ in (\ref{obs_2}) where $\mathbf A$ is the random  sampling matrix and   $\mathbf V_i = [\mathbf v_{i0}| \cdots| \mathbf v_{i(k-1)}]$ is the $N\times k$ matrix in which columns consist of the basis vectors of the subspace $\mathcal S_i$ for $i=0,1,\cdots,T-1$. The only term which depends on the sampling operator in the expression for the upper bound on the  probability of error in (\ref{Pe_bound0}) is $\lambda_{j\setminus i}$. When the sampling operator is a random projection matrix, $\lambda_{j\setminus i}$  can be evaluated as follows.
\begin{proposition}\label{prop_1}
Consider that  the sampling matrix $\mathbf A$ consists of elements drawn from a Gaussian ensemble with mean zero and variance 1. When $M-k$ is sufficiently large, we may approximate $\lambda_{j\setminus i}$ as
\begin{eqnarray*}
\lambda_{j\setminus i}\rightarrow  \frac{1}{\sigma_w^2}(M-k) || \sum_{m\in \mathcal W_{j\setminus i}} \mathbf v_{jm}  {c}_{j}(m)||_2^2
\end{eqnarray*}
where as defined before, $\mathcal W_{j\setminus i}$ ($l = |\mathcal W_{j\setminus i}|$) denotes the set consisting of indices of basis vectors in $\mathcal S_j$ which are not in $\mathcal S_i$.
\end{proposition}
\begin{proof}
See Appendix C.
\end{proof}

The quantity $\underset{{m\in \mathcal W_{j\setminus i}}}{\sum} \mathbf v_{jm}  {c}_{j}(m) $ is the portion of the original signal $\mathbf x$ that  is unaccounted for by the subspace $\mathcal S_i$ when the true subspace is $\mathcal S_j$ for $j\neq i$. Let $\tilde\alpha_{\min,l}^2 = \frac{1}{\sigma_w^2}\underset{i,j, j\neq i} {\min} ||\underset{{m\in \mathcal W_{j\setminus i}}}{\sum} \mathbf v_{jm}  {c}_{j}(m)  ||_2^2$ be the minimum (over $i,j = 0,1,\cdots, T-1$) SNR of the original signal $\mathbf x$ which is  unaccounted for by the subspace $\mathcal S_i$
 when the true subspace is $\mathcal S_j$ such  that $|\mathcal W_{j\setminus}| = l$ for $j\neq i$.  Then, with random sampling,  the upper bound on the  probability of error of the ML estimation  in (\ref{Pe_bound0}) reduces to (\ref{pe_ML}) after replacing $\alpha_{\min,l}^2 $ in (\ref{pe_ML}) by $\tilde\alpha_{\min,l}^2 $. It is worth mentioning that $\alpha_{\min,l}^2$ in (\ref{pe_ML}) is a measure of SNR after sampling while  $\tilde\alpha_{\min,l}^2 $ is a measure of SNR before sampling the signal.

\section{Subspace Recovery  from Structured Union of Subspaces}\label{structured}
In this section, we simplify the results obtained in Section \ref{ML_decoder}  when the subspaces in the union (\ref{union}) have structured properties leading to block sparsity.
\subsection{Block sparsity}\label{structured1}
With the block sparsity model as discussed in Subsection \ref{structured_block}, the observation vector $\mathbf y$ can be written in the form of
\begin{eqnarray}
\mathbf y = A V \mathbf c + \mathbf w = \mathbf B \mathbf c + \mathbf w \label{obs_block}
\end{eqnarray}
where $\mathbf B= AV$ is a $M\times N$ matrix, $V=[V_0|V_1|\cdots|V_{L-1}]$ is as defined in Subsection \ref{structured_block} and $\mathbf c$ has $L$ blocks (of size $d$ each)  in which all but $k_0$ blocks  are zeros; i.e., $\mathbf c$ is a block $k_0$-sparse vector. Further letting   $\mathbf B[i] = A V_i$ be a $M\times d$ matrix,  we can represent  $\mathbf B$ as a concatenation of column blocks $\mathbf B[i]$  for $i=0,1,\cdots, L-1$.  With this specific structure, the subspace recovery problem reduces to finding  the indices of blocks in $\mathbf c$ such that the elements inside that block are nonzero, i.e., the problem of  finding the block sparsity pattern. In addition to the  structured union of subspaces model considered here in which the block sparsity pattern is observed, there are other instances where  block sparsity arises such as in multiband signals \cite{Mishali1}, and in measurements of gene expression levels \cite{Eldar3}\cite{Parvaresh1}.

Define the support set of the block sparse signal $\mathbf c$ as
\begin{eqnarray*}
\mathcal U := \{ i \in \{0,1,\cdots, L-1\} | \mathbf c_i\neq 0\}
\end{eqnarray*}
which consists of the indices of the subspaces in the sum in (\ref{subspace_1}) or the indices of the  nonzero blocks in $\mathbf c$.
 With the above formulation, there are $ T = {L\choose k_0}$  such support sets and the $j$-th support set is denoted by $\mathcal U_j$ for $j=0,1,\cdots, T-1$.

Given that the true block support set is $\mathcal U_j$,
the measurement vector in (\ref{obs_block}) can be written as,
\begin{eqnarray*}
\mathbf y = \bar{\mathbf  B}_j \bar {\mathbf c}_j + \mathbf w
\end{eqnarray*}
where $\bar{\mathbf  B}_j = A \bar V_j$, $\bar V_j = [V_{u_j^0}| \cdots | V_{u_j^{k_0-1}}]$ and  $u_j^m$ denotes the $m$-th index in the set $\mathcal U_j$ for $m=0,1,\cdots, k_0-1$. Similar interpretation holds for the vector $\bar {\mathbf c}_j$.
To compute  the minimum number of samples  required for  asymptotically reliable  subspace recovery with this structured union of subspaces model based on ML estimation,  we can follow a similar approach as in Theorem \ref{lemma_M_conditions} with appropriate notation changes. In this case,  we can explicitly find $T_0(l)$ required in  Theorem \ref{lemma_M_conditions}. More specifically, for given $l$,  there are ${k_0\choose l}{{L-k_0}\choose l}$ number of sets such that $|\mathcal U_{j\setminus i}| = l$ for any given $\mathcal U_j$.  Then $T_j(l)=T_0(l)={k_0\choose l}{{L-k_0}\choose l}$. In the next section, we extend the analysis to the case where the sampling operator  is represented by random projections.

\subsection{Sampling via random projections}\label{random_block}
We assume  that the signal of interest $\mathbf x$ is a $N\times 1$ vector and the sampling  operator is a $M \times N$ matrix with random elements. Further, assume that the $N\times N$ basis matrix $V$ defined in Section \ref{structured_block} is orthonormal.

When the sampling operator is a $M\times N$ random matrix  $\mathbf A$,  the block sparse observation model in (\ref{obs_block}),  can be rewritten as,
\begin{eqnarray}
\mathbf y  =  \mathbf B \mathbf c + \mathbf w\label{obs_block_random}
\end{eqnarray}
where $\mathbf B = \mathbf A \mathbf V $, $\mathbf V$ is a $N\times N$ orthonormal matrix, $\mathbf c$ is a block sparse signal with $k_0$ nonzero blocks each of length  $d$  and  elements in $\mathbf A$ are drawn from a random ensemble.

 Compared to the analysis in Subsection \ref{random_sampling} with general unions when the sampling operator is a random projection matrix,  with the block sparsity model, we can further simplify the expression obtained  for $\lambda_{j\setminus i}$ in Proposition \ref{prop_1}. We define the minimum nonzero  block SNR as follows:
\begin{definition}
 The minimum nonzero block SNR  is defined as  $\mathrm{BSNR}_{\min} = \underset{m\in \mathcal U}{\min} \frac{|| \mathbf c_m ||_2^2}{\sigma_w^2}$ where $\mathcal U$ is the set containing the indices corresponding to nonzero blocks of the block sparse signal as defined in Section \ref{structured1}.
\end{definition}
\begin{proposition}\label{prop_4}
Let $BSNR_{\min}$ be the minimum nonzero block SNR of a block sparse signal. When the matrix $\mathbf A$ consists of elements drawn from a Gaussian ensemble with mean zero and variance 1,  for any $\mathcal U_j$ and $\mathcal U_i$ with $l = |\mathcal U_{j\setminus i}|$ we have,
\begin{eqnarray*}
 \lambda_{j\setminus i} =  \frac{1}{\sigma_w^2}(M-k_0d) \sum_{m=0}^{l-1}|| \mathbf V_{u_{j\setminus i}^m}  {\mathbf c}_{u_{j\setminus i}^m}||_2^2   \geq  {(M-k_0d)l} \mathrm{BSNR}_{\min}
\end{eqnarray*}
where $u_{j\setminus i}^m$ denotes the $m$-th index of the set $\mathcal U_{j\setminus i}$ which contains the indices of the subspaces in $\mathcal U_j$ which are not in $\mathcal U_i$.
\end{proposition}

\begin{proof}
Proof follows from Proposition \ref{prop_1} and the following results:
\begin{eqnarray}
||\sum_{m=0}^{l-1} \mathbf V_{u_{j\setminus i}^m}  {\mathbf c}_{u_{j\setminus i}^m}  ||_2^2 &=&  \langle \sum_{m=0}^{l-1} \mathbf V_{u_{j\setminus i}^m}  {\mathbf c}_{u_{j\setminus i}^m}, \sum_{m=0}^{l-1} \mathbf V_{u_{j\setminus i}^m}  {\mathbf c}_{u_{j\setminus i}^m} \rangle\nonumber\\
& = & \sum_{m=0}^{l-1} \langle  \mathbf V_{u_{j\setminus i}^m}  {\mathbf c}_{u_{j\setminus i}^m}, \mathbf V_{u_{j\setminus i}^m}  {\mathbf c}_{u_{j\setminus i}^m} \rangle  \nonumber\\
&+& \sum_{m\neq t} \langle  \mathbf V_{u_{j\setminus i}^m}  {\mathbf c}_{u_{j\setminus i}^m}, \mathbf V_{u_{j\setminus i}^t}  {\mathbf c}_{u_{j\setminus i}^t} \rangle\nonumber\\
&=& \sum_{m=0}^{l-1}|| \mathbf V_{u_{j\setminus i}^m}  {\mathbf c}_{u_{j\setminus i}^m}||_2^2\label{v_mc_m}
\end{eqnarray}
where the last equality is due to  the fact that the columns of $\mathbf V$ are orthogonal. Then (\ref{v_mc_m}) is lower bounded  by,
\begin{eqnarray*}
||\sum_{m=0}^{l-1} \mathbf V_{u_{j\setminus i}^m}  {\mathbf c}_{u_{j\setminus i}^m}  ||_2^2 \geq \sigma_w^2 l \mathrm{BSNR}_{\min}
\end{eqnarray*}
which completes  the proof.

\end{proof}

\begin{corollary}\label{corollary_3}
When the sampling operator is a random projection matrix where the elements are drawn from a Gaussian ensemble with mean zero and the variance $1$,  the upper bound on the  probability of error of  the ML estimation  in (\ref{Pe_bound0}) for block sparsity pattern recovery reduces  to,
\begin{eqnarray}
&~&P_e \leq  \sum_{l=1}^{k_0} {k_0 \choose l} {{L-k_0} \choose l} \nonumber\\
&~&\left( Q \left(\frac{1}{2}(1 - 2\eta_0 )\sqrt{ {(M-k)l} \mathrm{BSNR}_{\min}} \right)
+  \Psi\left({l}, \mathrm{BSNR}_{\min} \right)\right) \label{Pe_ML_random}
\end{eqnarray}
where $k_0=k/d$,  $\Psi\left({l}, \mathrm{BSNR}_{\min} \right) =  \frac{\sqrt 2}{2^{l}\Gamma(l/2)}(\eta_0{(M-k_0d)l} \mathrm{BSNR}_{\min} )^{l/2-1/2} \\ K_{l/2-1/2}(\eta_0{(M-k_0d)l} \mathrm{BSNR}_{\min}/2)$  and $0<\eta_0<1/2$.
\end{corollary}
Next, we investigate  sufficient conditions which state how the number of samples  $M$ scales with the other parameters $(L, k_0, d, \mathrm{BSNR}_{\min})$ to ensure that  the probability of error in (\ref{Pe_ML_random}) vanishes asymptotically with the block sparse model (\ref{obs_block_random}).

\begin{lemma}\label{lemma_3}
When $(M-k)\mathrm{BSNR}_{\min}\rightarrow \infty$, the probability of error of the ML estimation  (\ref{Pe_ML_random}) in recovering the block sparsity pattern vanishes asymptotically if the following conditions are satisfied:
\begin{eqnarray}
M > k +
\max\{\bar M_1, \bar M_2\} \label{M}
\end{eqnarray}
where \begin{eqnarray}
\bar M_1 = \frac{16}{BSNR_{\min} (1-2\eta_0)^2}(\log (L-k_0) + \log\left(\frac{e}{\sqrt 2}\right))
\end{eqnarray}
\begin{eqnarray}
\bar M_2 = \frac{4 (k_0/2 + r_0 -1)}{\eta_0r_0\mathrm{BSNR}_{\min}} \left\{\log(L-k_0) + \frac{1}{2} \log\left(\frac{2b_0e^2}{\sqrt \pi}\right)\right\}\label{M2}
\end{eqnarray}
with  $0 < \eta_0 < \frac{1}{2}$, $r_0 > 0$ and $b_0 =\frac{\sqrt{2\pi}}{4}$ are constants.
\end{lemma}

\begin{proof}
Proof follows from  Theorem \ref{lemma_M_conditions} and using the relations, that ${k_0\choose l}\leq {{L-k_0}\choose l} $ for $k_0\leq L/2$, and $\log ({{L-k_0}\choose l}) \leq l \log \left(\frac{e(L-k_0)}{l}\right)$.

\end{proof}

From Lemma \ref{lemma_3}, we can write the minimum  number of random samples required for reliable  block sparsity pattern recovery asymptotically in the form of $\mathcal O(k + \frac{\eta_4}{ BSNR_{\min}}\log(L-k_0))$  for some constant $\eta_4$   in the case where $k_0$ is sufficiently small.

\begin{remarks}
When $BSNR_{\min}\rightarrow \infty$,  $M > k$ measurements are sufficient for asymptotically  reliable block sparsity pattern recovery   with ML estimation.
\end{remarks}

\subsection{Revisiting the standard sparsity model }
In the standard sparsity model considered widely in the CS literature, the subspaces in the union (\ref{union}) are assumed  to be $k$-dimensional subspaces of an orthonormal basis. This is a special case of the block sparse model when  $d=1$.    To have a fair comparison to  the performance of the ML estimation in the presence of noise with the standard sparsity model and block sparsity model, we introduce further notations. Define the minimum component SNR, $\mathrm{CSNR}_{\min} = \underset{m\in \mathcal U, i=0,\cdots,d-1}{\min} \frac{|| \mathbf c_m(i) ||_2^2}{\sigma_w^2}$ so that $BSNR_{\min} \geq d CSNR_{\min}$.  Then, when the sampling is performed via random projections, the probability of error of the ML estimation  with the standard sparsity model is upper bounded as in (\ref{Pe_ML_random_standard})
\begin{figure*}
\begin{eqnarray}
P_e \leq  \sum_{l=1}^{k} {k\choose l} {{N-k} \choose l}\left( Q \left(\frac{1}{2}(1 - 2\eta_0 )\sqrt{ {(M-k)l} \mathrm{CSNR}_{\min}} \right) +  \Psi\left({l}, \mathrm{CSNR}_{\min} \right)\right) \label{Pe_ML_random_standard}
\end{eqnarray}
\end{figure*}
where $ \Psi\left({l}, \mathrm{CSNR}_{\min} \right)$ is as defined in Corollary \ref{corollary_3}.
With these notations, the probability of error of the ML estimation  with block sparsity model (\ref{Pe_ML_random}) can be rewritten as in (\ref{Pe_ML_random_0})
\begin{figure*}
\begin{eqnarray}
P_e \leq  \sum_{l=1}^{k_0} {k_0 \choose l} {{L-k_0} \choose l}\left( Q \left(\frac{1}{2}(1 - 2\eta_0 )\sqrt{ {(M-k)l} d\mathrm{CSNR}_{\min}} \right) +  \Psi\left({l}, d\mathrm{ CSNR}_{\min} \right)\right)\label{Pe_ML_random_0}
\end{eqnarray}
\end{figure*}
where $L = N/d$ and $k_0 = k/d$ as defined previously.
By obtaining the conditions under which   $P_e$ in (\ref{Pe_ML_random_standard}) and (\ref{Pe_ML_random_0}) vanishes asymptotically, it can be shown that the dominant part of the required number of random samples for reliable subspace recovery  asymptotically in the presence of noise  can be expressed in the form of $\mathcal O(k + \frac{1}{d}\frac{\hat {\eta_1}}{CSNR_{\min}}\log(L-k_0))$ with block sparsity model and $\mathcal O(k + \frac{\hat \eta_2}{ CSNR_{\min}}(\log(N-k) ))$ with the standard      sparsity model where $ \hat {\eta}_1$ and $\hat {\eta}_2$ are positive constants. Thus, when the signal $\mathbf x$ exhibits block sparsity pattern with $k= k_0 d$  where $k$ is the  total number of non zero coefficients of the  sparse signal, $k_0$ is the number of blocks and $d$ is the block size,  the required number of random samples beyond $k$ (i.e. in  terms of $M-k$) is reduced  by approximately a factor of $d$  compared to that with the standard sparsity model.
Note  that the above analysis is for the worst case, i.e. the upper bounds on the probability of error  are obtained considering the minimum block/component SNR. The actual number of measurements required for reliable subspace recovery  can be less than that  predicted  in Lemma \ref{lemma_3}.

\section{Comparison with  Existing Results}\label{discussion}

\subsection{Existing results for support recovery  with the standard sparsity model }
The most related existing work on deriving sufficient conditions for the  ML estimation   to succeed in the presence of noise with the standard sparsity model is presented in \cite{wain2}. There, taking the canonical basis as the sparsifying basis,  the results are derived based on the following bound on  the probability of error:
\begin{eqnarray}
P_e \leq \sum_{l=1}^k {k\choose l}{{N-k}\choose l} 4 \exp \left\{-\frac{(M-k) l CSNR_{\min}}{64(l CSNR_{\min} + 8)}\right\}. \label{Pe_wain}
\end{eqnarray}
When  $CSNR_{\min}\rightarrow \infty$, it can be easily seen that this upper bound is bounded away from zero (i.e. it is bounded by $4 e^{-(M-k)/64} \left( {Ld\choose k} - 1 \right) > 0$).
Based on the upper bound (\ref{Pe_wain}),  it was shown in \cite{wain2} that
\begin{eqnarray}
M &>& k + (\eta_1  + 2048 )\max \left\{\tilde M_1=\log\left({N-k \choose k}\right)\right., \nonumber\\
&~&\left.\tilde M_2=\frac{\log(N - k)}{CSNR_{\min}}\right\}\label{M_wain}
 \end{eqnarray}measurements are required for asymptotically  reliable sparsity pattern recovery  where $\eta_1$ is a constant (which is different from the one used earlier in the paper). When the minimum  component SNR, $CSNR_{\min} \rightarrow \infty$, the ML estimation  requires $k+(\eta_1 + 2048) k\log((N-k)/k)$ measurements for asymptotically  reliable recovery, which is much larger than $k$. However, as shown in \cite{Baron1,Fletcher1},  when the measurement noise power is negligible (or in the no noise case), the exhaustive search  decoder  is capable of recovering the sparsity pattern with $M =k+1$  measurements with high probability. Thus, the limits predicted by the existing results in the literature  for sparsity pattern recovery in terms of the minimum number of measurements  show a gap with  what is actually required. When $d=1$,  $\mathbf V$ is the standard canonical basis, and $\mathbf A$ is a random Gaussian matrix, the structured union of subspaces model  considered in Section \ref{structured} (specifically the equation (\ref{obs_block_random})) is the same as the model considered in \cite{wain2}. Our results show that when $CSNR_{\min}\rightarrow \infty$, the upper bound on the probability of error in (\ref{Pe_ML_random_standard}) vanishes  with the standard sparsity model when $M > k$.  More specifically, when $CSNR_{\min}\rightarrow \infty$, our results imply  that $\mathcal O(k)$ measurements are sufficient for  asymptotically reliable   sparsity pattern recovery with the ML estimation which is intuitive. Further, at finite $CSNR_{\min}$,  when $\tilde M_2$ dominates  $\tilde M_1$ in (\ref{M_wain}) the lower bound  in \cite{wain2} has the same scaling with respect to $L$, $k$, $d$ and $CSNR_{\min}$ to that is obtained in this paper  with the standard sparsity model.

\subsection{Existing results for signal recovery with  union of subspaces}
The problem of stable recovery of signals that lie in a union of subspaces model is addressed  in \cite{Eldar1,Baraniuk4,Lu1,Blumensath1,Duarte2}.  In these works, the main focuss is to derive sufficient conditions  that ensure reliable recovery of the complete   signals while in this paper, our focus is only in identifying the low dimensional subspace in which  the signal lies.
Nevertheless, it is interesting to compare the results since it will provide insights into identifying the regions of the parameters ($L$, $k$, $SNR$, etc..) that ensure asymptotically  reliable subspace recovery  using  the existing algorithms developed for exact signal recovery.   

The following result is shown in  \cite{Blumensath1}.
\begin{theorem}[\cite{Blumensath1}]\label{theorem_Blu}
For any given $t>0$, if
\begin{eqnarray}
M > \frac{2}{c \delta} \left(\log(2T) + k \log\left(\frac{12}{\delta}\right)  + t \right)\label{M_min}
\end{eqnarray}
then, the matrix $A$ in (\ref{obs_1}) satisfies the restricted isometry property (RIP) with the restricted isometry  constant $\delta$ (for formal definition of RIP readers may refer to \cite{Blumensath1}).
\end{theorem}
In \cite{Eldar3}, the authors  derived the sufficient  conditions for complete signal recovery in the block sparsity model.  When the samples are acquired via random projections (elements in $\mathbf A$ are Gaussian)  with  the notations used in Section \ref{random_block}, the minimum number of samples required for the sampling matrix to satisfy  block RIP  with high probability   is given by (from Theorem \ref{theorem_Blu} and \cite{Eldar3})
\begin{eqnarray}
M \geq \frac{36}{7 \delta} \left(\log\left(2 {L\choose k_0 }\right) + k \log\left(\frac{12}{\delta}\right) + t\right)\label{M_Eldar}
\end{eqnarray}
for some $t>0$ and $0 < \delta < 1$ is the restricted isometry constant. This is roughly in the order of $\tilde{\eta}_1 k+\tilde{\eta}_2 k_0\log (L/k_0)$ for some positive constants $\tilde{\eta}_1$ and $\tilde{\eta}_2$. Thus, block sparse signals can be reliably recovered using computationally tractable algorithms (e.g. extension of BP - mixed  $l_2/l_1$  norm recovery algorithms)  with $\tilde{\eta}_1 k+\tilde{\eta}_2 k_0\log (L/k_0)$ measurements when there is no noise. In the presence of noise, the BP based   algorithm  developed in \cite{Eldar1} is shown to be robust so  that the norm of the recovery error
is bounded by the noise level.  As shown in Section \ref{random_block},   it requires roughly the order of $k+ ( \eta_4 /BSNR_{\min} )\log(L-k_0)$ measurements (when $k_0$ is fairly small) for  reliable block sparsity pattern recovery with ML estimation.   Here, the second term is significant at finite $BSNR_{\min}$ while it vanishes when  $BSNR_{\min}\rightarrow \infty$. At finite $BSNR_{\min}$, when $k_0$ is sublinear w.r.t. $L$, it can be shown that $ k_0\log (L/k_0) >> \log(L-k_0)$. Thus, in that region of $k_0$, the relevant scaling obtained in (\ref{M_Eldar}) is larger than   what is required by the optimal ML estimation  derived in this paper at finite $BSNR_{\min}$. The exact difference between them depends on the  value of $BSNR_{\min}$ and the relevant constants.

\section{Numerical Results}\label{numerical}
Several computationally tractable algorithms  for sparsity pattern recovery with standard sparsity  have been  derived and discussed quite extensively in the literature.
Extensions of  such algorithms for  model based  or structured CS have  also been considered in several recent works. For example, extensions of CoSamp and iterative hard thresholding algorithms for model based CS were considered in \cite{Baraniuk4}. Extensions of OMP algorithm for block sparsity pattern recovery (BOMP) were considered in \cite{Eldar3,Fang1} while \cite{Eldar1,Lv1,Friedman1} considered the Group Lasso algorithm  for block sparse signal recovery.

Our goal in this section is to validate the tightness of the derived upper bounds on the probability of error of the ML estimation   and  provide numerical results to illustrate the performance gap  when employing  practical algorithms for subspace recovery.    Simulating the ML algorithm is difficult due to its high computational complexity in the high dimensions.  Nevertheless, we show the performance for reasonably sized signal dimensions and samples just to demonstrate the tightness of the probability of error bound. For the structured union of subspaces model considered in Section \ref{structured1}, the problem reduces to recovering  the block sparsity pattern of a block sparse signal.
The performance of the ML algorithm is compared to  block-OMP as proposed in \cite{Eldar3} which is provided  in Algorithm \ref{algo1} where the set $\hat{\mathcal U}$ contains the estimated indices of the nonzero blocks of a  block sparse signal.
\begin{algorithm}
Input: $\mathbf y$, $\mathbf B$, $k_0$
\begin{enumerate}
\item Initialize $t=1$, $\hat{\mathcal U}(0) = \emptyset $, residual vector $\mathbf r_{0} = \mathbf y$
\item Find the index $\lambda(t)$ such that
$
\lambda(t) = \underset{i = 0,\cdots,L-1}{\arg~ \max} ~||\mathbf B[i]^{*} \mathbf r_{t-1}||_2
$
 \item  Set $\hat{\mathcal U}(t) = \hat{\mathcal U}(t-1) \cup \{\lambda (t)\}$
\item  Compute the projection operator $\mathbf P(t) = \mathbf B(\hat{\mathcal U}(t)) \left( \mathbf B(\hat{\mathcal U}(t)) ^T  \mathbf B(\hat{\mathcal U}(t)) \right)^{-1}  \mathbf B(\hat{\mathcal U}(t)) ^T$. Update the residual vector:  $\mathbf r_{t} = (\mathbf I - \mathbf P(t))\mathbf y$ (note: $\mathbf B(\hat{\mathcal U}(t))$ denotes the submatrix of $\mathbf B$ in which columns are taken from $\mathbf B$ corresponding to the indices in $\hat{\mathcal U}(t)$)
     \item  Increment $t=t+1$ and go to step 2 if $t\leq k_0$, otherwise, stop and set $\hat{\mathcal U} = \hat{\mathcal U}(t-1)$

 \end{enumerate}
 \caption{Block-OMP (B-OMP) for block sparsity pattern recovery }\label{algo1}
 \end{algorithm}

\begin{figure}
    \centering
     \subfigure[$N=50$, $L=25$,  $d=2$, $BSNR_{\min}=13 dB$ ]
    {
        \includegraphics[width=3.50in]{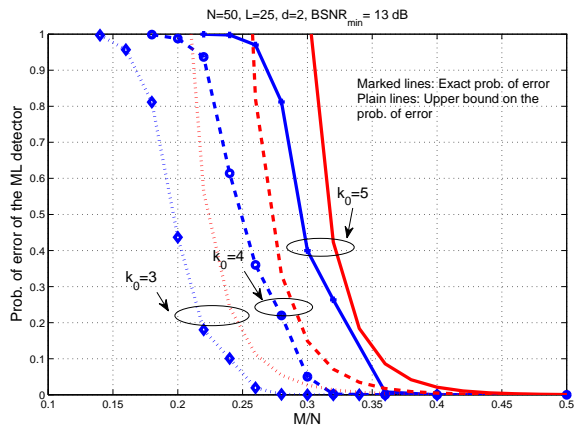}

    }
    \subfigure[$N=L=50$, $d=1$, $CSNR_{\min}=10 dB$ ]
    {
        \includegraphics[width=3.50in]{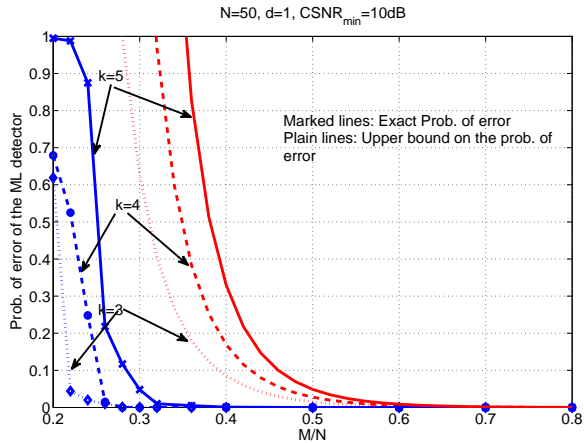}

    }
    \caption{Exact probability of error and the derived upper bound on the probability of error of  the ML recovery  for  block sparsity pattern recovery }
    \label{fig_0}
\end{figure}

 Results in Figures \ref{fig_0} and \ref{fig_1},  are based on the  special structure as considered in  (\ref{subspace_addition}) for subspaces  leading to block sparsity and the sampling operator is assumed to be a random matrix in which elements are drawn from a Gaussian ensemble with mean zero and variance $1$. Further, we let $N\times N$ matrix $\mathbf V$ be the standard canonical basis.  In Fig \ref{fig_0} (a), the exact  probability of error of the ML estimation  (obtained via simulation) and the upper bound on the probability of error derived in (\ref{Pe_ML_random}) vs $M/N$  are shown. In the block sparsity model, we let $N=50$, $d=2$, $L=25$, $BSNR_{\min} =13 dB$ and three different plots correspond to $k_0=3,4,5$. In Fig. \ref{fig_0} (b), we let $d=1$ (i.e. the standard sparsity model)  so that the upper bound on the  probability of error reduces to  (\ref{Pe_ML_random_standard}). We also let $CSNR_{\min} =10 dB$ and different curves correspond to different values of $k$ in Fig.\ref{fig_0} (b).   The exact probability of error of the ML estimation is obtained via Monte Carlo simulations with $10^5$ runs. In the upper bounds (\ref{Pe_ML_random}) and (\ref{Pe_ML_random_standard}), we let $\eta_0 = 1/4$. It can be seen from Fig. \ref{fig_0}(a) and \ref{fig_0}(b)  that the derived upper bound on the probability of error is fairly a tight bound on the exact  probability of error especially as  $M/N$ increases and the tightness is more significant in Fig. \ref{fig_0}(a). It should be  noted that for $d=2$, we have $k=k_0 d$, thus the total number of non zero coefficients is larger in Fig. \ref{fig_0}(a) than that with $d=1$ in Fig. \ref{fig_0}(b). Thus, it is seen that derived upper bound becomes tighter as $k$ increases.  It is also worth mentioning that the derived upper bound on the probability of error in \cite{wain2} with the standard sparsity model (as in (\ref{Pe_wain})) is  bounded away from $1$ for the selected parameter values mentioned above.

\begin{figure}
\centerline{\epsfig{figure=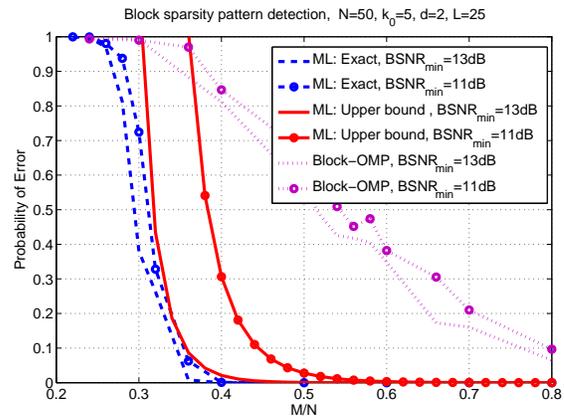,width=8.5cm}}
\caption{Performance of the ML estimation  and the B-OMP algorithm for  block sparsity pattern recovery; $L=25$, $k_0=5$, $d=2$, and thus $k=10$,    $N=50$}
  \label{fig_1}
\end{figure}

In Fig.  \ref{fig_1}, the performance of the block sparsity pattern recovery with  ML and B-OMP algorithms  is shown when  $BSNR_{\min}$ varies. In Fig.  \ref{fig_1}, we let $k_0=5$, $L=25$, $d=2$ and  $N=50$.  For B-OMP, $10^4$ runs are performed for a given projection matrix and averaged over $100$ runs. In Fig. \ref{fig_1}, the ratio between the minimum and maximum block SNR in both cases considered  is set at $1.825$.  As observed in Fig. \ref{fig_0}, from Fig. \ref{fig_1} it can be seen that the derived upper bound on the  probability of error of the ML estimation  is fairly closer to the exact probability error obtained via Monte Carlo simulations, especially as $BSNR_{\min}$ increases. Further,  for a  given finite $BSNR_{\min}$, there seems to be  a considerable  performance gap between  the B-OMP and  the ML estimation. That is the price to pay for the computational complexity of  the ML estimation  vs the computationally efficient B-OMP algorithm.

\section{Conclusion}\label{conclusion}
In this paper, we investigated the problem of subspace recovery  based on reduced dimensional samples  when the signal of interest lies in a union of subspaces.    With a given sampling operator, we derived the performance  of the optimal ML estimation  for subspace recovery  in the presence of noise  in terms of the probability of error. We further  obtained  conditions under which   asymptotically  reliable subspace recovery is  guaranteed.

We extended the analysis   to a special case of union of subspaces model which reduces to block sparsity.  When the samples are obtained via random projections, sufficient conditions required for asymptotically reliable block sparsity pattern recovery with the ML estimation  were derived. Performance gain in terms of the minimum number of samples required for asymptotically reliable  subspace recovery  with the block sparse model  was quantified compared to that with the standard sparsity model.    Our results further strengthen the existing results for  sparsity pattern recovery  with the standard  sparsity model used in CS framework with random projections. More specifically, our results for sufficient conditions for asymptotically reliable   subspace recovery  are derived based on a tighter bound on the probability of error of the ML estimation  compared to the existing results in  the literature with the standard sparsity model.
We further discussed and illustrated numerically the performance gap between the ML estimation  and the computationally tractable algorithms (e.g. B-OMP) used for subspace recovery  with the structured  union of subspaces model.

An interesting future direction will be to extend the analysis with the  single node system to  a multiple node system in distributed networks.
\section*{Appendix A}
\subsection*{Proof of Lemma \ref{lammaU_1}}
To prove Lemma \ref{lammaU_1}, we consider a similar argument to that  considered in  \cite{wain2} with certain differences as noted in the following.
As shown in \cite{wain2}, we may write,
\begin{eqnarray*}
\Delta_{ij}(\mathbf y) = ||\mathbf P_i^\bot \mathbf y||_2^2 - ||\mathbf P_i^\bot \mathbf w ||_2^2 + ||\mathbf P_i^\bot \mathbf w||_2^2 - ||\mathbf P_j^\bot \mathbf y||_2^2.
\end{eqnarray*}
For any given $\delta > 0$, define the events
\begin{eqnarray}
h_1(\delta) = \left\{ | \frac{||\mathbf P_j^\bot \mathbf y||_2^2 - ||\mathbf P_i^\bot \mathbf w ||_2^2}{\sigma_w^2}|\geq \delta \right\}\label{h_1}
\end{eqnarray}
and
\begin{eqnarray}
h_2(\delta) = \left\{  \frac{||\mathbf P_i^\bot \mathbf y||_2^2 - ||\mathbf P_i^\bot \mathbf w ||_2^2}{\sigma_w^2 }\leq 2 \delta \right\}. \label{h_2}
\end{eqnarray}
Then $Pr(\Delta_{ij}(\mathbf y) < 0)$ implies that at least one event in (\ref{h_1}) and (\ref{h_2}) is true. Based on  the union bound, we can write
\begin{eqnarray*}
Pr(\Delta_{ij}(\mathbf y) < 0) \leq Pr(h_1(\delta) ) + Pr(h_2(\delta) ).
\end{eqnarray*}

With the standard sparsity model and assuming that the sampling is performed via random projections, upper bounds on the probabilities  $Pr(h_1(\delta)) $ and $Pr(h_2(\delta)) $ are derived in \cite{wain2}. In contrast, in the following, we derive exact value for $Pr(h_2(\delta)) $ and a tighter bound for $Pr(h_1(\delta)) $ assuming that the sampling operator $A$ is known. Thus, even for  the standard sparsity model, the results presented in this paper tightens the results derived in \cite{wain2}.

We first evaluate $Pr(h_1(\delta)) $.  Let $\Delta_{ij}^1(\mathbf y) = \frac{1}{\sigma_w^2} ( ||\mathbf P_j^\bot \mathbf y||_2^2 - ||\mathbf P_i^\bot \mathbf w ||_2^2)$. Assuming the true subspace is $\mathcal S_j$, $\Delta_{ij}^1(\mathbf y)$ reduces to $\Delta_{ij}^1(\mathbf y)  = \frac{1}{\sigma_w^2} ( ||\mathbf P_j^\bot \mathbf w
||_2^2 - ||\mathbf P_i^\bot \mathbf w ||_2^2)$. As shown in \cite{wain2},  the random variable $\Delta_{ij}^1(\mathbf y)$ can be represented as $\Delta_{ij}^1(\mathbf y) = x_1 - x_2$ where $x_1$ and $x_2$ are independent and $x_1, x_2 \sim \mathcal X_{l}^2$ where $l$ is the cardinality of the set  $\mathcal W_{j\setminus i}$ as defined before.  With these notations, we can write
\begin{eqnarray*}
Pr(h_1(\delta)) &=& Pr (|x_1-x_2| \geq \delta)\nonumber \\
&=&  Pr ((x_1-x_2) \geq  \delta) + Pr ((x_1-x_2) < -\delta ).
\end{eqnarray*}
The pdf of the random variable $w = x_1 -x_2$ is symmetric around zero and thus we have,
\begin{eqnarray*}
Pr(h_1(\delta)) = 2 Pr ((x_1-x_2) \geq \delta).
\end{eqnarray*}

\begin{proposition}
When $x_1 \sim \mathcal X_l^2$ and $x_2 \sim \mathcal X_l^2$, the random variable $w = x_1 - x_2$ has the following pdf:
\begin{eqnarray}
&~&f_w(w) \nonumber\\
&=& \left\{
\begin{array}{ccc}
f_w^+(w) = \frac{w^{\frac{l}{2} - \frac{1}{2}}}{\sqrt \pi 2^ l \Gamma(l/2)}  K_{1/2 - l/2 } \left(\frac{w}{2}\right); ~ &  if ~ w \geq 0\\
f_w^-(w) = \frac{(-w)^{\frac{l}{2} - \frac{1}{2}}}{\sqrt \pi 2^ l \Gamma(l/2)}  K_{1/2 - l/2 } \left(\frac{-w}{2}\right); ~ &  if ~ w < 0
\end{array}\right.\label{f_w}
\end{eqnarray}
where $ K_\nu (x)$ is the modified Bessel function.
\end{proposition}
\begin{proof}
Since $x_1$ and $x_2$ are independent, the pdf of $w = x_1 - x_2 $ is given by \cite{Papoulis1}
\begin{eqnarray*}
f_w(w) = \left\{
\begin{array}{ccc}
\int_0^{\infty} f_{x_1}( w + x_2) f_{x_2} (x_2) d x_2; ~ if~ w \geq 0\\
\int_{-w}^{\infty} f_{x_1}( w + x_2) f_{x_2} (x_2) d x_2; ~ if~ w < 0
\end{array}\right.
\end{eqnarray*}
First consider the case where $w > 0$. Then
\begin{eqnarray*}
f_w^+(w) &=& \int_0^{\infty} \frac{(w+x_2)^{l/2-1} e^{-(w + x_2)/2}}{2^{l/2}\Gamma(l/2)}  \frac{x_2^{l/2-1} e^{- x_2/2}}{2^{l/2}\Gamma(l/2)} dx_2\nonumber\\
&=& \frac{e^{-w/2}}{2^l (\Gamma(l/2))^2} \int_0^{\infty} x_2^{l/2-1} (w+x_2)^{l/2 -1} e^{-x_2} dx_2\nonumber\\
&=& \frac{e^{-w/2}}{2^l (\Gamma(l/2))^2} \frac{1}{\sqrt \pi} w^{l/2-1/2} e^{w/2} \Gamma(l/2) K_{1/2-l/2}(w/2)\nonumber\\
&=& \frac{w^{l/2-1/2} K_{1/2-l/2}(w/2)}  {\sqrt\pi 2^l \Gamma(l/2)}
\end{eqnarray*}
where $ K_\nu (x)$ is the modified Bessel function and the third equality is obtained using the integral result $\int_0^{\infty} x^{\nu-1} (x+\beta)^{\nu-1} e^{-\mu x} dx = \frac{1}{\sqrt \pi} \left(\frac{\beta}{\mu}\right)^{\nu-1/2} e^{\beta \mu/2} \Gamma(\nu) K_{1/2 - \nu} \left(\frac{\beta \mu}{2}\right)$ for $\mu,\nu >0$ in \cite[p.~348]{Gradshteyn1}.

When $w < 0$, we have,
\begin{eqnarray}
f_w^-(w) =  \frac{e^{-w/2}}{2^l (\Gamma(l/2))^2} \int_{-w}^{\infty} x_2^{l/2-1} (w+x_2)^{l/2 -1} e^{-x_2} dx_2. \label{fw_minus}
\end{eqnarray}
Letting $z=-w$ where $z > 0$, (\ref{fw_minus}) can be rewritten as,
\begin{eqnarray}
f_w^-(w) =  \frac{e^{z/2}}{2^l (\Gamma(l/2))^2} \int_{z}^{\infty} x_2^{l/2-1} (x_2-z)^{l/2 -1} e^{-x_2} dx_2.
\end{eqnarray}
Using the integral result, $\int_u^{\infty} x^{\nu-1} (x-u)^{\nu-1} e^{-\mu x} dx = \frac{1}{\sqrt \pi} \left(\frac{u}{\mu}\right)^{\nu-1/2} e^{-\mu u /2} \Gamma(\nu) K_{\nu - 1/2 } \left(\frac{ \mu u}{2}\right)$ in \cite[p.~347]{Gradshteyn1} and the relation $K_{\nu}(x) = K_{-\nu}(x)$, we get $f_w^-(w)$ as in  (\ref{f_w}), completing the proof.
\end{proof}

\begin{proposition}
For $\delta > 0$, the probability $Pr(w > \delta)$ is given by,
\begin{eqnarray*}
Pr(w > \delta) &\leq& \frac{\sqrt 2}{2^{l+1}\Gamma(l/2)} \delta^{l/2-1/2} K_{l/2-1/2}(\delta/2)
\end{eqnarray*}
where $ K_\nu (x)$ is the modified Bessel function, and $\Gamma(.)$ is the Gamma function.
\end{proposition}

\begin{proof}
Based on (\ref{f_w}), we have
\begin{eqnarray}
Pr(w > \delta) &=&  \int_{\delta}^{\infty} f_w^+(w) dw \nonumber\\
&=& \int_{\delta}^{\infty} \frac{w^{l/2-1/2}  K_{1/2-l/2}(w/2)}  {\sqrt\pi 2^l \Gamma(l/2)} dw. \label{p_w}
\end{eqnarray}
Using the equivalent integral representation of $ K_\nu(az) = \frac{z^\nu}{2} \int_0^{\infty} e^{-\frac{a}{2} \left(t+\frac{z^2}{t}\right)} t^{-\nu -1} dt$ \cite[p.~917]{Gradshteyn1}, we can write the integral in (\ref{p_w}) as,
\begin{eqnarray}
&~&Pr(w > \delta) \nonumber\\
&=& \frac{1}{\sqrt\pi 2^{l+1} \Gamma(l/2)} \int_{\delta}^\infty \int_0^\infty  e^{-\frac{1}{4} \left(t+\frac{w^2}{t}\right)} t^{l/2 -3/2} dt dw. \label{p_w2}
\end{eqnarray}
Since  $\int_{\delta}^{\infty} e^{-\frac{w^2}{4t}} dw = \sqrt{2\pi} Q\left(\frac{\delta}{\sqrt{2t}}\right)$, (\ref{p_w2}) reduces  to,
\begin{eqnarray}
Pr(w > \delta)  &=& \frac{\sqrt 2}{2^{l+1} \Gamma(l/2)} \int_{0}^\infty e^{-t/4} t^{l/2  -3/2} Q\left(\frac{\delta}{\sqrt{2t}}\right) dt\nonumber
\end{eqnarray}
\begin{eqnarray}
&\leq & \frac{\sqrt 2}{2^{l+2} \Gamma(l/2)} \int_{0}^\infty t^{l/2  -3/2} e^{-t/4-\frac{\delta^2}{4t}}  dt\label{bound1}\\
&=& \frac{\sqrt 2}{2^{l+1}\Gamma(l/2)} \delta^{l/2-1/2} K_{l/2-1/2}(\delta/2)\label{bound2}
\end{eqnarray}
where we used the inequality $Q(x) \leq \frac{1}{2}e^{-\frac{x^2}{2}}$ for $x>0$, and the relation, $\int_0^{\infty} x^{\nu-1} e^{-\beta/x - \gamma x } dx = 2\left(\frac{\beta}{\gamma}\right)^{\nu/2} K_{\nu} (2\sqrt{\beta\gamma})$ for $\beta > 0$ and $\gamma >0$ \cite[p.~368]{Gradshteyn1} while obtaining (\ref{bound1}) and (\ref{bound2}), respectively, which completes the proof.
\end{proof}
 Then, we have
\begin{eqnarray}
Pr(h_1(\delta)) = \frac{\sqrt 2}{2^{l}\Gamma(l/2)} \delta^{l/2-1/2} K_{l/2-1/2}(\delta/2).
\end{eqnarray}

Next we compute the quantity $Pr (h_2(\delta))$. Let $\Delta_{ij}^2(\mathbf y) = \frac{1}{\sigma_w^2} (||\mathbf P_i^\bot \mathbf y||_2^2 - ||\mathbf P_i^\bot \mathbf w ||_2^2)$.
Then we have,
\begin{eqnarray*}
\Delta_{ij}^2(\mathbf y)= \frac{1}{\sigma_w^2} (||\mathbf P_i^\bot  {\mathbf B}_{j\setminus i}{\mathbf c}_{j\setminus i}||_2^2 + 2 \mathbf w^T \mathbf P_i^\bot  {\mathbf B}_{j\setminus i}{\mathbf c}_{j\setminus i}).
\end{eqnarray*}
Since $\mathbf w \sim \mathcal N(\mathbf 0, \sigma_w^2 \mathbf I_M)$, $\Delta_{ij}^2(\mathbf y)$ is a Gaussian random variable with  pdf,
\begin{eqnarray*}
\Delta_{ij}^2(\mathbf y) \sim \mathcal N \left( \frac{1}{\sigma_w^2} ||\mathbf P_i^\bot  {\mathbf B}_{j\setminus i}{\mathbf c}_{j\setminus i}||_2^2, \frac{4}{\sigma_w^2} ||\mathbf P_i^\bot  {\mathbf B}_{j\setminus i} {\mathbf c}_{j\setminus i}||_2^2\right).
\end{eqnarray*}
Thus,
\begin{eqnarray*}
Pr (h_2(\delta) ) &=& Pr \left({\Delta_{ij}^2(\mathbf y)} \leq 2 \delta \right)\nonumber\\
&=&1 - Q\left(\frac{2 \delta  - \frac{1}{\sigma_w^2} ||\mathbf P_i^\bot  {\mathbf B}_{j\setminus i}{\mathbf c}_{j\setminus i}||_2^2}{ \frac{2}{\sigma_w} ||\mathbf P_i^\bot  {\mathbf B}_{j\setminus i}{\mathbf c}_{j\setminus i}||_2}\right)\nonumber\\
&=& 1 - Q\left(\frac{2 \delta -\lambda_{j\setminus i}}{  2 \sqrt{\lambda_{j\setminus i}}}\right).
\end{eqnarray*}
Since it is desired to  control $\delta$ such that $Pr (h_2(\delta) ) \leq 1/2$, we select $\delta^* = \eta_0 \lambda_{j\setminus i}$ where $\eta_0 < \frac{1}{2}$. With this choice $Pr (h_2(\delta) )$ reduces to,
\begin{eqnarray*}
Pr (h_2(\delta) )  = Q\left(\frac{1}{2}\sqrt{\lambda_{j\setminus i}} (1 - 2 \eta_0 )\right)
\end{eqnarray*}
where we used the relation $1- Q(-x) = Q(x)$ for $x > 0$, while  $Pr(h_1(\delta))$ reduces to,
\begin{eqnarray}
Pr(h_1(\delta)) = \frac{\sqrt 2}{2^{l}\Gamma(l/2)} (\eta_0 \lambda_{j\setminus i})^{l/2-1/2} K_{l/2-1/2}(\eta_0 \lambda_{j\setminus i}/2).
\end{eqnarray}

\section*{Appendix B}
\subsection*{Proof of Theorem \ref{lemma_M_conditions}}
To obtain  conditions under which the probability of error bound in  (\ref{pe_ML}) asymptotically vanishes, we rely on the following corollary.

\begin{corollary}
Let $T_0(l) $ and $\alpha_{\min,l}^2$ be as defined in Subsection \ref{eval_lamda}. The probability of error of the ML estimation  in (\ref{pe_ML}) is further upper bounded by
\begin{eqnarray}
P_e \leq  \sum_{l=1}^{k} T_0(l) \left( \frac{1}{2} e^{ -\frac{1}{8} (1 - 2\eta_0 )^2 (M-k)\alpha_{\min,l}^2} +  \phi_{l}\right)\label{Pe_ML_bound}
\end{eqnarray}
where
\begin{eqnarray}
\phi_{l} =
\frac{ \sqrt{2\pi}}{4\Gamma(l/2)} \left(\frac{1}{4} \eta_0 (M-k)\alpha_{\min,l}^2\right)^{l/2 -1}e^{-\frac{1}{2} \eta_0 (M-k)\alpha_{\min,l}^2}
\label{phi_s_bound}
\end{eqnarray}
when $(M-k) \alpha_{\min,l}^2 >> (l/2 -1/2)$ for all $l=1,2,\cdots,k$ and  $0<\eta_0 < 1/2$.
\end{corollary}

\begin{proof}
Using the Chernoff bound for the $Q$ function where $Q(x) \leq \frac{1}{2} e^{-\frac{x^2}{2}}$, we can upper bound the  term $Q \left(\frac{1}{2}(1 - 2\eta_0 )\sqrt{(M-k) \alpha_{\min,l}^2} \right)$ as,
\begin{eqnarray*}
Q \left(\frac{1}{2}(1 - 2\eta_0 )\sqrt{(M-k) \alpha_{\min,l}^2} \right) &\leq& \frac{1}{2} e^{ -\frac{1}{8} (1 - 2\eta_0 )^2 (M-k)\alpha_{\min,l}^2}
\end{eqnarray*}
for $\eta_0 <\frac{1}{2}$.

To obtain (\ref{phi_s_bound}) we used the relation $K_{\nu}(z) \approx \sqrt{\frac{\pi}{2 z}} e^{-z}$ when $\nu << z$,  completing the proof.
\end{proof}
It is further noted that when $k$ is fairly small and $\alpha_{\min,l}^2$ is sufficiently large, the condition required for (\ref{phi_s_bound}) is often satisfied.
We consider the conditions under which the each term in (\ref{Pe_ML_bound}) goes to  $0$ asymptotically, equivalently logarithm of each term $\rightarrow -\infty$. First consider the first term in the summation  in (\ref{Pe_ML_bound}) for which the logarithm gives,
\begin{eqnarray*}
&~&\log T_0(l) + \log(1/2) - \frac{1}{8}(1-2\eta_0)^2 (M-k) \alpha_{\min,l}^2 \nonumber\\
&\leq& \underset{l}{\max}\left\{ \log(T_0(l))  +  \log(1/2)\right.\nonumber\\
 &-&\left. \frac{1}{8}(1-2\eta_0)^2 (M-k)~   \{\alpha_{\min,l}^2 \}\right\}
\rightarrow -\infty
\end{eqnarray*}
as $(M-k)\rightarrow \infty$ when $M > k + M_1$ where $M_1=\underset{l=1,\cdots,k}{\max} \left\{\frac{8}{(1-2\eta_0)^2\alpha_{\min,l}^2}
 \left\{\log(T_0(l)) + \log(1/2)\right\}\right\}$.
Considering the second term in (\ref{Pe_ML_bound}),
let
\begin{eqnarray}
\Pi_1 & =& \log T_0(l) + \log\left(\frac{b_0}{\Gamma(l/2)}\right) \nonumber\\
&+& (l/2-1) \log\left(\frac{1}{4}\eta_0 (M-k) \alpha_{\min,l}^2\right)\nonumber\\
 &-& \frac{1}{2}\eta_0 (M-k) \alpha_{\min,l}^2\label{bound_M_4}
\end{eqnarray}
where $b_0 = \frac{\sqrt{2\pi}}{4}$.
When $\frac{1}{4}\eta_0 (M-k) \alpha_{\min,l}^2$ is sufficiently large, we can find $0 < q_0 < \frac{1}{(k/2 -1)}$ such that $\log\left(\frac{1}{4}\eta_0 (M-k) \alpha_{\min,l}^2\right) < q_0 \frac{1}{2}\eta_0 (M-k) \alpha_{\min,l}^2$. Then (\ref{bound_M_4}) is  upper bounded by
\begin{eqnarray}
\Pi_1 &\leq& \underset{l=1,\cdots,k}{\max}\left\{ \log(T_0(l))  + \log\left(\frac{b_0}{\Gamma(3/2)}\right) \right.\nonumber\\
&-&\left. \left(\frac{1}{2}\eta_0 (M-k) \alpha_{\min,l}^2\right) ( 1 -  q_0 (k/2 -1))\right\}= \Pi_2\label{bound_M_4_2}
\end{eqnarray}
where $0 < q_0 < \frac{1}{(k/2 -1)}$. We can write $q_0$ in the form of  $ q_0 = \frac{1}{2(k/2 + r_0 -1)}$ for some $r_0 > 0$. Thus, (\ref{bound_M_4_2}) can be rewritten as
\begin{eqnarray*}
\Pi_2 &=& \underset{l=1,\cdots,k}{\max}\left\{ \log(T_0(l)) + \log\left(\frac{2b_0}{\sqrt\pi}\right) \right. \nonumber\\
&-& \left. \left(\frac{1}{2}\eta_0 (M-k) \alpha_{\min,l}^2\right) \frac{r_0}{r_0 + k/2 -1}\right \}\rightarrow - \infty
\end{eqnarray*}
as $(M-k)\rightarrow \infty$ when $M > k+ M_2$ where $M_2 = \underset{l=1,\cdots,k}{\max} \left\{\frac{2(k/2 + r_0 -1)}{r_0\eta_0   \alpha_{\min,l}^2}
 \left\{\log(T_0(l)) + \log\left(\frac{2b_0}{\sqrt\pi}\right)\right\} \right\}$,  $0<\eta_0 < 1/2 $, $b_0 = \frac{\sqrt{2\pi}}{4}$,  and $r_0 >0$.

\section*{Appendix C}
\subsection*{Proof of Proposition  \ref{prop_1}}
We rewrite $ \lambda_{j\setminus i} =   \frac{1}{\sigma_w^2} || {\mathbf P}_i^\bot {\mathbf B}_{j\setminus i}{\mathbf c}_{j\setminus i}||_2^2 $.
The $t$-th element of the vector $ {\mathbf B}_{j\setminus i}{\mathbf c}_{j\setminus i}$ can be written as $\langle \mathbf a_t, \underset{{m\in \mathcal W_{j\setminus i}}}{\sum} \mathbf v_{jm}  {c}_{j}(m)  \rangle$  where $\mathbf a_t$'s are row vectors of $\mathbf A$ for $t=0,1,\cdots, M-1$. Assuming that the  elements of $\mathbf A$ are  independent Gaussian with mean zero and variance 1, it can be easily seen that  $\langle \mathbf a_t, \underset{{m\in \mathcal W_{j\setminus i}}}{\sum} \mathbf v_{jm}  {c}_{j}(m)  \rangle$ is a realization of a  Gaussian random variable  with mean zero and  variance $||\underset{{m\in \mathcal W_{j\setminus i}}}{\sum} \mathbf v_{jm}  {c}_{j}(m)  ||_2^2$. Further, the elements of $ {\mathbf B}_{j\setminus i}{\mathbf c}_{j\setminus i}$ are independent of each other since $\mathbf a_t$'s are independent for $t=0,1,\cdots, M-1$. Thus,  the random vector $ {\mathbf B}_{j\setminus i}{\mathbf c}_{j\setminus i} \sim \mathcal N (\mathbf 0, ||\underset{{m\in \mathcal W_{j\setminus i}}}{\sum} \mathbf v_{jm}  {c}_{j}(m)  ||_2^2\mathbf I_M)$.  With given realizations, consider again the transformation $\mathbf Q_i^T {\mathbf B}_{j\setminus i}{\mathbf c}_{j\setminus i}$ where $\mathbf Q_i$ is the unitary matrix with eigenvectors of $ {\mathbf P}_i^\bot$. Since the elements in ${\mathbf B}_{j\setminus i}{\mathbf c}_{j\setminus i}$ are independent and identically distributed (iid), the unitary transformation  does not change the distribution of ${\mathbf B}_{j\setminus i}{\mathbf c}_{j\setminus i}$.  Then $||  {\mathbf P}_i^{\bot}  {\mathbf B}_{j\setminus i}{\mathbf c}_{j\setminus i}||_2^2 = ||\Lambda_i \mathbf Q_i^T  {\mathbf B}_{j\setminus i}{\mathbf c}_{j\setminus i}||_2^2$   is a sum of $M-k $ iid random variables. Thus  when $(M-k)$ is sufficiently large, invoking the law of large numbers, we may approximate  $ ||{\mathbf P}_i^\bot {\mathbf B}_{j\setminus i}{\mathbf c}_{j\setminus i}||_2^2 $ $   \rightarrow (M- k) ||\underset{{m\in \mathcal W_{j\setminus i}}}{\sum} \mathbf v_{jm}  {c}_{j}(m)  ||_2^2$ which completes the proof.

\bibliographystyle{IEEEtran}
\bibliography{IEEEabrv,bib1}

\begin{thebibliography}{10}
\providecommand{\url}[1]{#1}
\csname url@samestyle\endcsname
\providecommand{\newblock}{\relax}
\providecommand{\bibinfo}[2]{#2}
\providecommand{\BIBentrySTDinterwordspacing}{\spaceskip=0pt\relax}
\providecommand{\BIBentryALTinterwordstretchfactor}{4}
\providecommand{\BIBentryALTinterwordspacing}{\spaceskip=\fontdimen2\font plus
\BIBentryALTinterwordstretchfactor\fontdimen3\font minus
  \fontdimen4\font\relax}
\providecommand{\BIBforeignlanguage}[2]{{%
\expandafter\ifx\csname l@#1\endcsname\relax
\typeout{** WARNING: IEEEtran.bst: No hyphenation pattern has been}%
\typeout{** loaded for the language `#1'. Using the pattern for}%
\typeout{** the default language instead.}%
\else
\language=\csname l@#1\endcsname
\fi
#2}}
\providecommand{\BIBdecl}{\relax}
\BIBdecl

\bibitem{candes1}
E.~Cand$\grave{e}$s, J.~Romberg, and T.~Tao, ``Robust uncertainty principles:
  exact signal reconstruction from highly incomplete frequency information,''
  \emph{IEEE Trans. Inform. Theory}, vol.~52, no.~2, pp. 489 -- 509, Feb. 2006.

\bibitem{donoho1}
D.~Donoho, ``Compressed sensing,'' \emph{IEEE Trans. Inform. Theory}, vol.~52,
  no.~4, pp. 1289--1306, Apr. 2006.

\bibitem{candes2}
E.~Cand$\grave{e}$s and T.~Tao, ``Near-optimal signal recovery from random
  projections: Universal encoding strategies?'' \emph{IEEE Trans. Info.
  Theory}, vol.~52, no.~12, pp. 5406 -- 5425, Dec. 2006.

\bibitem{Eldar_B1}
Y.~C. Eldar and G.~Kutyniok, \emph{Compressed Sensing: Theory and
  Applications}.\hskip 1em plus 0.5em minus 0.4em\relax Cambridge University
  Press, 2012.

\bibitem{Malioutov1}
D.~Malioutov, M.~Cetin, and A.Willsky, ``A sparse signal reconstruction
  perspective for source localization with sensor arrays,'' \emph{IEEE Trans.
  Signal Processing}, vol.~53, no.~8, pp. 3010--3022, Aug. 2005.

\bibitem{Cevher1}
V.~Cevher, P.~Indyk, C.~Hegde, and R.~G. Baraniuk, ``Recovery of clustered
  sparse signals from compressive measurements,'' in \emph{Int. Conf. Sampling
  Theory and Applications (SAMPTA 2009)}, Marseille, France, May. 2009, pp.
  18--22.

\bibitem{Natarajan1}
B.~K. Natarajan, ``Sparse approximate solutions to linear systems,'' \emph{SIAM
  J. Computing}, vol.~24, no.~2, pp. 227--234, 1995.

\bibitem{Miller1}
A.~J. Miller, \emph{Subset Selection in Regression}.\hskip 1em plus 0.5em minus
  0.4em\relax New York, NY: Chapman-Hall, 1990.

\bibitem{Larsson1}
E.~G. Larsson and Y.~Selén, ``Linear regression with a sparse parameter
  vector,'' \emph{IEEE Trans. Signal Processing}, vol.~55, no.~2, pp. 451--460,
  Feb. 2007.

\bibitem{Tian1}
Z.~Tian and G.~Giannakis, ``Compressed sensing for wideband cognitive radios,''
  in \emph{Proc. Acoust., Speech, Signal Processing (ICASSP)}, Honolulu, HI,
  Apr. 2007, pp. IV--1357--IV--1360.

\bibitem{Mishali2}
M.~Mishali and Y.~C. Eldar, ``Wideband spectrum sensing at sub-nyquist rates,''
  \emph{IEEE Signal Processing Magazine}, vol.~28, no.~4, pp. 102--135, July
  2011.

\bibitem{Mishali3}
M.~Mishali, Y.~C. Eldar, O.~Dounaevsky, and E.~Shoshan, ``\textsc{X}ampling:
  Analog to digital at sub-nyquist rates,'' \emph{IET Circuits, Devices and
  Systems}, vol.~5, no.~1, pp. 8--20, Jan. 2011.

\bibitem{SSChen1}
S.~S. Chen, D.~L. Donoho, and M.~A. Saunders, ``Atomic decomposition by basis
  pursuit,'' \emph{SIAM J. Sci. Computing}, vol.~20, no.~1, pp. 33--61, 1998.

\bibitem{wain2}
M.~J. Wainwright, ``Information-theoretic limits on sparsity recovery in the
  high-dimensional and noisy setting,'' \emph{IEEE Trans. Inform. Theory},
  vol.~55, no.~12, pp. 5728--5741, Dec. 2009.

\bibitem{Wainwright4}
------, ``Sharp thresholds for high-dimensional and noisy sparsity recovery
  using $l_1$-constrained quadratic programming (lasso),'' \emph{IEEE Trans.
  Inform. Theory}, vol.~55, no.~5, pp. 2183--2202, May 2009.

\bibitem{wang5}
W.~Wang, M.~J. Wainwright, and K.~Ramachandran, ``Information-theoretic limits
  on sparse signal recovery: Dense versus sparse measurement matrices,''
  \emph{IEEE Trans. Inform. Theory}, vol.~56, no.~6, pp. 2967--2979, Jun. 2010.

\bibitem{Fletcher1}
A.~K. Fletcher, S.~Rangan, and V.~K. Goyal, ``Necessary and sufficient
  conditions for sparsity pattern recovery,'' \emph{IEEE Trans. Inform.
  Theory}, vol.~55, no.~12, pp. 5758--5772, Dec. 2009.

\bibitem{Akcakaya1}
M.~M.~Akcakaya and V.~Tarokh, ``Shannon-theoretic limits on noisy compressive
  sampling,'' \emph{IEEE Trans. Inform. Theory}, vol.~56, no.~1, pp. 492--504,
  Jan. 2010.

\bibitem{Reeves1}
G.~Reeves and M.~Gastpar, ``Sampling bounds for sparse support recovery in the
  presence of noise,'' in \emph{IEEE Int. Symp. on Information Theory (ISIT)},
  Toronto, ON, Jul. 2008, pp. 2187--2191.

\bibitem{tang1}
G.~Tang and A.~Nehorai, ``Performance analysis for sparse support recovery,''
  \emph{IEEE Trans. Inform. Theory}, vol.~56, no.~3, pp. 1383--1399, March
  2010.

\bibitem{goyal1}
V.~K. Goyal, A.~K. Fletcher, and S.~Rangan, ``Compressive sampling and lossy
  compression,'' \emph{IEEE Signal Processing Magazine}, vol.~25, no.~2, pp.
  48--56, Mar. 2008.

\bibitem{thakshilaj4}
T.~Wimalajeewa and P.~K. Varshney, ``Performance bounds for sparsity pattern
  recovery with quantized noisy random projections,'' \emph{IEEE Journal of
  Selected Topics in Signal Processing, Special Issue on Robust Measures and
  Tests Using Sparse Data for Detection and Estimation}, vol.~6, no.~1, pp. 43
  -- 57, Feb. 2012.

\bibitem{Baraniuk4}
R.~G. Baraniuk, V.~Cevher, M.~Duarte, and C.Hegde, ``Model based compressed
  sensing,'' \emph{IEEE Trans. Information Theory}, vol.~56, no.~4, pp.
  1982--2001, Apr. 2010.

\bibitem{Lu1}
Y.~M. Lu and M.~N. Do, ``A theory for sampling signals from a union of
  subspaces,'' \emph{IEEE Trans. Signal Processing}, vol.~56, no.~6, pp.
  2334--2345, June 2008.

\bibitem{Eldar1}
Y.~C. Eldar and M.~Mishali, ``Robust recovery of signals from a structured
  union of subspaces,'' \emph{IEEE Trans. Information Theory}, vol.~55, no.~11,
  pp. 5302--5316, Nov. 2009.

\bibitem{Blumensath1}
T.~Blumensath and M.~E. Davies, ``Sampling theorems for signals from the union
  of finite-dimensional linear subspaces,'' \emph{IEEE Trans. Inform. Theory},
  vol.~55, no.~4, pp. 1872--1882, Apr. 2009.

\bibitem{Eldar3}
Y.~C. Eldar, P.~Kuppinger, and H.~Bolcskei, ``Block-sparse signals: Uncertainty
  relations and efficient recovery,'' \emph{IEEE Trans. Signal Processing},
  vol.~58, no.~6, pp. 3042--3054, June 2010.

\bibitem{Duarte2}
M.~Duarte and Y.~C. Eldar, ``Structured compressed sensing: From theory to
  applications,'' \emph{IEEE Trans. Signal Processing}, vol.~59, no.~9, pp.
  4053--4085, Sep. 2011.

\bibitem{Bruchstein1}
A.~M. Bruchstein, T.~J. Shan, and T.~Kailath, ``The resolution of overlapping
  echos,'' \emph{IEEE Trans. Acoust., Speech, and Signal Process.}, vol.~33,
  no.~6, pp. 1357--1367, Dec. 1985.

\bibitem{Gedalyahu1}
K.~Gedalyahu and Y.~C. Eldar, ``Time-delay estimation from low-rate samples: A
  union of subspaces approach,'' \emph{IEEE Trans. Signal Processing}, vol.~58,
  no.~6, pp. 3017--3031, June 2010.

\bibitem{Ben1}
Z.~Ben-Haim, T.~Michaeli, and Y.~C. Eldar, ``Performance bounds and design
  criteria for estimating finite rate of innovation signals,'' \emph{IEEE
  Trans. Information Theory}, vol.~58, no.~8, pp. 4993--5015, Aug. 2012.

\bibitem{Dragotti1}
P.~Dragotti, M.~Vetterli, and T.~Blu, ``Sampling moments and reconstructing
  signals of finite rate of innovation: Shannon meets strang-fix,'' \emph{IEEE
  Trans. Signal Processing}, vol.~55, no.~5, pp. 1741--1757, May 2007.

\bibitem{Haim2}
Z.~Ben-Haim and Y.~C. Eldar, ``Near-oracle performance of greedy block-sparse
  estimation techniques from noisy measurements,'' \emph{IEEE Journal of
  Selected Topics in Signal Processing}, vol.~5, no.~5, pp. 1032--1047, Sept.
  2011.

\bibitem{Mishali1}
M.~Mishali and Y.~Eldar, ``Blind multi-band signal reconstruction: Compressed
  sensing for analog signals,'' \emph{IEEE Trans. Signal Processing}, vol.~57,
  no.~3, pp. 993--1009, Mar. 2009.

\bibitem{Parvaresh1}
F.~Parvaresh, H.~Vikalo, S.~Misra, and B.~Hassibi, ``Recovering sparse signals
  using sparse measurement matrices in compressed dna microarrays,'' \emph{IEEE
  Journal of Selected Topics in Signal Processing}, vol.~2, no.~3, pp.
  275--285, June 2008.

\bibitem{Baron1}
D.~Baron, M.~B. Wakin, M.~F. Duarte, S.~Sarvotham, and R.~G. Baraniuk,
  ``Distributed compressive sensing,'' \emph{Rice Univ. Dept. Elect. Comput.
  Eng. Houston, TX, Tech. Rep. TREE–0612}, Nov 2006.

\bibitem{Fang1}
J.~Fang and H.~Li, ``Block-sparsity pattern recovery from noisy observations,''
  in \emph{Proc. Acoust., Speech, Signal Processing (ICASSP)}, Mar. 2012, pp.
  3321--3324.

\bibitem{Lv1}
X.~Lv, G.~Bi, and C.~Wan, ``The group lasso for stable recovery of block-sparse
  signal representations,'' \emph{IEEE Trans. Signal Processing}, vol.~59,
  no.~4, pp. 1371--1382, Apr. 2011.

\bibitem{Friedman1}
J.~Friedman, T.~Hastie, and R.~Tibshirani, ``A note on the group lasso and a
  sparse group lasso,'' \emph{[Online] Available:
  http://arxiv.org/pdf/1001.0736, preprint}, 2010.

\bibitem{Papoulis1}
A.~Papoulis and S.~U. Pillai, \emph{Probability, Random Variables and
  Stochastic Processes}.\hskip 1em plus 0.5em minus 0.4em\relax McGraw Hill,
  4th Edition, 2002.

\bibitem{Gradshteyn1}
I.~S. Gradshteyn and I.~M. Ryzhik, \emph{Table of Integrals, Series and
  Products}.\hskip 1em plus 0.5em minus 0.4em\relax Elseveir Academic Press,
  2007.

\end{thebibliography}
\end{document}